\documentclass[11pt,draftcls,onecolumn]{IEEEtran}
\usepackage{amsfonts,epsfig,amsmath,latexsym,amssymb,amscd,multirow,graphicx,lscape,amsmath,amssymb,bm,pifont,graphicx,amssymb,amsmath}
\usepackage{bbm}
\usepackage[]{cite}
\usepackage{undertilde}
\usepackage[usenames,dvipsnames]{xcolor}
\usepackage{graphicx}
\usepackage{lineno}  
\usepackage[normalem]{ulem} 
\usepackage{dsfont}

\def\argmin{\operatornamewithlimits{arg\,min}}

\newcommand{\model}{{\cal M}}

\newcommand{\lone}{{ \ell}_1}

\newcommand{\lqreg}{{ \ell}_q}

\newcommand{\vect}[1] {\mbox{${\bf #1}$}}

\newcommand{\reel}{{\mathbb R}}

\usepackage{framed}
\usepackage{setspace}

\newcommand{\Data}{{\bf y}}

\newcommand{\state}{{\bm \theta}}
\newcommand{\State}{{\bm \theta}}

\newcommand{\Normal}{{\cal N}}
\newcommand{\ComSpace}{E}

\newcommand{\Expec}{\mathbb{E}}
\newcommand{\Var}{\mathbb{V}\text{ar}}

\newcommand{\ESS}{\mathbb{ESS}}
\newcommand{\CESS}{\mathbb{CESS}}

\newcommand{\Kernel}{{\cal K}}
\newcommand{\BackKernel}{{\cal L}}

\newcommand{\Target}{{\pi}}
\newcommand{\CoolSchedule}{\phi}
\newcommand{\UnNorTarget}{{\gamma}}
\newcommand{\NormConst}{Z}
\newcommand{\PropDist}{\eta}

\newcommand{\PostDist}{p}
\newcommand{\TestFunction}{\varphi}

\newcommand{\NbIterSMC}{T}

\newcommand{\NbParticles}{N}

\newcommand{\ISWeight}{W}
\newcommand{\NormISWeight}{\widetilde{W}}
\newcommand{\IncreStep}{\varrho}
\newcommand{\IncreWeight}{{{w}}}
\newcommand{\UnNorIncreWeight}{{w}}

\newcommand{\ParametricParameter}{\tau}

\newcommand{\DimObs}{{n_{\Data}}}
\newcommand{\DimState}{{n_{\State}}}
\newcommand{\OneData}{{ y}}

\newcommand{\AcceptRatio}{\alpha}

\newcommand{\MainTarget}{\pi}

\newcommand{\Uniform}{{\cal U}}

\newcommand{\ParsIndexSMC}{m}

\newcommand{\NbMCMCmove}{N_{\text{MCMC}}}

\newcommand{\Block}{{\boldsymbol\varrho}}

\newcommand{\TimeSet}{{\cal N}}

\newcommand{\Coefficient}{{\boldsymbol \beta}}
\newcommand{\OneCoefficient}{{\beta}}
\newcommand{\NbCoefficient}{p}

\newcommand{\OneMeanY}{\mu}

\newcommand{\Covariates}{{\vect{x}}}

\newcommand{\DesignMatrix}{{\boldsymbol \Phi}}
\newcommand{\designmatrix}{{\Phi}}

\newcommand{\RowMatrix}{i}
\newcommand{\ColumnMatrix}{j}

\definecolor{orange}{rgb}{1,0.5,0}
\definecolor{violet}{rgb}{0.53, 0.0, 0.69}
\definecolor{pakistangreen}{rgb}{0.0, 0.4, 0.0}

\usepackage{multicol}

\usepackage[font=bf]{subfig}

\RequirePackage{dsfont,mathrsfs}
\usepackage{eqparbox}
\usepackage{algorithmic}
\usepackage{algorithm}
\makeatletter

\def\cleardoublepage{\clearpage\if@twoside \ifodd\c@page\else%
  \hbox{}%
  \thispagestyle{empty}
  \newpage%
  \if@twocolumn\hbox{}\newpage\fi\fi\fi}

\makeatother
 
{%

\hrulefill
\vspace*{0.5cm}%
\end{minipage}
}

\usepackage{multirow}
{ \begin{list}%
	{$\bullet$}%
	{\setlength{\labelwidth}{25pt}%
	 \setlength{\leftmargin}{30pt}%
	 \setlength{\itemsep}{\parsep}}}%
{ \end{list} }

\renewcommand{\epsilon}{\varepsilon}

\usepackage{bm}
\usepackage{undertilde}
\newcommand{\iid}{\stackrel{ \mathrm{iid}}{\sim}}
\newcommand{\StateT}{\widetilde{\State}}
\newcommand{\WeightESS}{{ \textnormal{\tiny \textsc{ESS}}}}

\newcommand{\WeightDemix}{{ \textnormal{\tiny DeMix}}}
\newcommand{\TimeIndex}{t}
\renewcommand{\ParametricParameter}{\gamma}

\newcommand{\Correct}[1]{{{#1}}}

\newtheorem{proposition}{\noindent \textbf{Proposition}}

\usepackage[left=2.2cm,right=2.2cm,top=2.4cm,bottom=2.3cm]{geometry}

\begin{document}
\title{Efficient Sequential Monte-Carlo Samplers for Bayesian Inference}
\author{Thi Le Thu Nguyen$^{1}$, Fran\c{c}ois~Septier$^{1}$, Gareth W.~Peters$^{2}$ and Yves Delignon$^{1}$
\begin{center}
{\footnotesize {\ 
\textit{
$^{1}$ Institut Mines-T\'el\'ecom / T\'el\'ecom Lille / CRIStAL UMR CNRS 9189,  Villeneuve d'ascq, France.\\
$^{2}$ Department of Statistical Sciences, University College London (UCL), London, England. \\
} } }
\end{center}
}

\maketitle

\vspace{-2cm} 
\begin{abstract}
%
In many problems, complex non-Gaussian and/or nonlinear models are required to accurately describe a physical system of interest. In such cases, Monte Carlo algorithms are remarkably flexible and extremely powerful approaches to solve such inference problems. However, in the presence of a high-dimensional and/or multimodal posterior distribution, it is widely documented that standard Monte-Carlo techniques could lead to poor performance. In this paper, the study is focused on a Sequential Monte-Carlo (SMC) sampler framework, a more robust and efficient Monte Carlo algorithm. Although this approach presents many advantages over traditional Monte-Carlo methods, the potential of this emergent technique is however largely underexploited in signal processing. In this work, we aim at proposing some novel strategies that will improve the efficiency and facilitate practical implementation of the SMC sampler specifically for signal processing applications. Firstly, we propose an automatic and adaptive strategy that selects the sequence of distributions within the SMC sampler that minimizes the asymptotic variance of the estimator of the posterior normalization constant. This is critical for performing model selection in modelling applications in Bayesian signal processing. The second original contribution we present improves the global efficiency of the SMC sampler by introducing a novel correction mechanism that allows the use of the particles generated through all the iterations of the algorithm (instead of only particles from the last iteration). This is a significant contribution as it removes the need to discard a large portion of the samples obtained, as is standard in standard SMC methods. This will improve estimation performance in practical settings where computational budget is important to consider.

\end{abstract}

\section{Introduction} \label{Introduction}
Bayesian inference is an important area of signal processing, relevant to a wide range of real-world applications. As opposed to the point estimators (means, variances) used by classical statistics, Bayesian statistics is concerned with generating the posterior distribution of the unknown parameters given both the data and some prior density for these parameters. As such, Bayesian statistics provides a much more complete picture of the uncertainty in the estimation of the unknown parameters of a model.
		
		In Bayesian inference, the model parameters are regarded as random variables, and the main object of interest is the posterior distribution, i.e. the distribution of the parameters given the data. Specifically, the posterior density is defined via Bayes' Theorem as the normalized product of the prior density and the likelihood
		\begin{equation}\label{postdis}
			\PostDist(\State|\Data)=\dfrac{p(\Data|\State)p(\State)}{\int_{\ComSpace} p(\Data|\State)p(\State) d\State}  \varpropto p(\Data|\State)p(\State) ,
		\end{equation}
		where $\ComSpace$ denotes the parameter space of $\State$. From equation (\ref{postdis}) it is clear that $\PostDist(\State|\Data)$   involves a contribution from the observed data through $p(\Data|\State)$, and a contribution from prior information quantified through $p(\State)$. The posterior $\PostDist(\State|\Data)$ contains all relevant information on the unknown parameters $\State$ given the observed data $\Data$. All statistical inference can be deduced from the posterior distribution, through appropriate choice of summaries. This typically takes the form of evaluating integrals such as,
	\begin{equation}
		J=\int \varphi(\State) \PostDist(\State|\Data) d\State ,
	\end{equation}
	for some integrable function $\varphi(\State)$ with respect to the posterior distribution. For example, point estimates for unknown parameters are given by the posterior means, i.e., $\varphi(\State)=\State;$ prediction for future data $\tilde{\Data}$ is based on the posterior predictive distribution $p(\tilde{\Data}|\State,\Data)$, both of which are easily expressed as integral functionals of the posterior.
	
	There are several potential difficulties in any practical implementation of a Bayesian method. One of them is the issue of specifying the prior distribution. However, extra difficulties arise in actually calculating the various quantities required.  First, in applying Bayes's theorem we need to compute the integral in the denominator. Secondly, the process of inference may require the calculation of further integrals of other operations on the posterior distribution. These calculations may be difficult to perform in practice, especially in complex problems involving a high-dimensional and/or multimodal posterior distribution. These integrals are typically approximated using Monte Carlo methods, requiring the ability to sample from general probability distributions which can typically be only evaluated up to a normalizing constant.

\subsection{Existing works}		
	
	In many cases, using standard sampling techniques such as inversion or rejection to sample from a target distribution (\Correct{i.e.} posterior distribution) is either not possible or can prove to be too much of a computational burden. This has led to the development in recent years of much more advanced algorithms which allow one to obtain the required samples from the target distribution. Standard approaches are mostly based on Markov chain Monte Carlo (MCMC), where the equilibrium distribution of the chain is the target distribution and its ergodic mean converges to the expected value \cite{Robert2004}. MCMC algorithms have been applied with success to many problems, e.g. \cite{septier2011,Djuric2002,Dobigeon_IEEE_SP_Mag_2014,Doucet2005SignalPro}. However, there are two major drawbacks with MCMC methods.  Firstly, it is difficult to assess when the Markov chain has reached its stationary regime of interest. Secondly, if the target distribution is highly multi-modal, MCMC algorithms can easily become trapped in local modes.
	
	 In recent years, more robust and efficient Monte Carlo algorithms have been established in order to efficiently explore high dimensional and multimodal spaces. Many of them are population based, in that they deal explicitly with a collection of samples at each iteration, including  population-based MCMC \cite{LiangWong2001,Jasra2007} and sequential Monte-Carlo samplers \cite{Peters:2005ti,DelMoral2006,peters2012sequential}. In \cite{Jasra2007}, the authors provide a detailed review as well as several illustrations showing that such population strategies can lead to significant improvement compared to standard MCMC techniques.
	
	Population-based MCMC was originally developed by Geyer \cite{Geyer1991}. Further advances came with an evolutionary Monte Carlo algorithm in \cite{Liang2000} and \cite{LiangWong2001}, each of which attempted to produce genetic algorithm type moves to improve the mixing of the Markov chain. It works by simulating a population of several Markov chains with different invariant distributions in parallel using MCMC. The population is updated by mutation (Metropolis update in one single chain), crossover (partial state swapping between different chains), and exchange operators (full state swapping between different chains). However, like standard MCMC, this population-based MCMC algorithm still suffers from the difficulty of assessing when the Markov chains have reached their stationary regime. 
	
	The second population-based simulation approach is the sequential Monte Carlo sampler proposed in \cite{Peters:2005ti,DelMoral2006}. Sequential Monte Carlo (SMC) methods is a class of sampling algorithms which combine importance sampling and resampling. They have been primarily used in the ``particle filter'' setting to solve optimal filtering problems; see, for example, \cite{Cappe2007} and \cite{Doucet2009} for recent reviews. In this context, SMC methods/particle filters have enjoyed wide-spread use in various applications (tracking, computer vision, digital communications) due to the fact that they provide a simple way of approximating complex filtering distribution sequentially in time. But in \cite{Peters:2005ti,DelMoral2006,peters2012sequential}, the authors developed  a general framework that allows SMC to be used to simulate from a single and static target distribution, thus becoming a promising alternative to standard MCMC methods.  The SMC sampler framework involves the construction of  a sequence \Correct{of} artificial distributions on spaces of increasing dimensions which admit the distributions of interests as particular marginals. The mechanism is similar to sequential importance sampling (resampling) (\cite{Lui2008} and \cite{Doucet2001}), with one of the crucial differences being the framework under which the particles are allowed to move, resulting in differences in the calculation of the weights of the particles. 
	
	These methods have several advantages over population-based MCMC methods. Firstly, unlike MCMC, SMC methods do not require any burn-in period and do not face the sometimes contentious issue of diagnosing convergence of a Markov chain. Secondly, as discussed in \cite{Jasra2007}, compared to population-based MCMC, SMC samplers are  a richer class of methods since there is substantially more freedom in specifying  the mutation kernels in SMC: kernels \Correct{do not need} to be reversible or even Markov (and hence can be time adaptive). Finally, unlike MCMC, SMC samplers provide an unbiased estimate of the unknown normalizing constant of the posterior distribution whatever the number of particles used \cite{DelMoral2000}.
	
\subsection{Contributions}	
	
	Although this approach presents many advantages over traditional  MCMC methods, the potential of these emergent techniques is however largely underexploited in signal processing. In this paper, we therefore focus our study on this technique by aiming at proposing some novel strategies that will improve the efficiency and facilitate practical implementation of the SMC sampler. More specifically, we firstly derive some convergence results of the SMC sampler for some specific choice of the backward kernel, which is generally used  in practice, as well as under a perfectly mixing forward kernel. This convergence result facilitates the analysis of the SMC sampler and in particular highlights the impact of the choice of the sequence of target distributions on the algorithm performance. The first contribution of this paper consists in proposing a adaptive strategy in order to obtain an automatic choice of the sequence of intermediate target distributions that optimizes the asymptotic variance of the estimator of the marginal likelihood. The second main contribution is the derivation of effective schemes in order to improve the global efficiency of the SMC sampler. The idea developed in this paper is to propose some correction mechanisms that allow the use of the particles generated through all the iterations of the algorithm (instead of only the particles from the last iteration) in order to improve the accuracy of the empirical approximation of the target distribution.

\section{SMC Samplers for Bayesian inference} \label{GeneralSMC}

\subsection{General Principle of SMC Samplers}\label{MethodologySMCSamplers}

Sequential Monte Carlo (SMC) methods are a class of sampling algorithms which combine importance sampling and resampling. The SMC sampler is based on two main ideas:
\begin{enumerate}
\item[a)] Rather than sampling directly the complex distribution of interest, a sequence of intermediate target distributions, $\left\{\Target _{\TimeIndex }\right\}_{\TimeIndex=1}^\NbIterSMC$, are designed, that transitions smoothly from a simpler distribution to the one of interest.
In Bayesian inference problems, the target distribution is the posterior $\Target_{T}(\state  )=p(\state|{\bm z})$, thus a natural choice for such a sequence of intermediate distributions is  to select the following \cite{Neal2001}
\begin{equation}\label{SequenceSMCsampler}
\Target _{\TimeIndex }(\state  ) = \frac{\UnNorTarget_t(\state)}{\NormConst_t}\varpropto p(\state  ) p (\Data | \state  )^{\CoolSchedule _{\TimeIndex }} ,
\end{equation} 
where $\left \{   \CoolSchedule _{\TimeIndex }\right\} $ is a non-decreasing temperature schedule with $\CoolSchedule _{0}=0$ and $\CoolSchedule _{\NbIterSMC }=1$ and $\UnNorTarget_t(\state)$ corresponds to the unnormalized target distribution $\left( \text{i.e. } \UnNorTarget_t(\state)=p(\state  ) p (\Data | \state  )^{\CoolSchedule _{\TimeIndex }}\right)$ and $\NormConst_t = \int_{\Theta}p(\state  )  p (\Data | \state  )^{\CoolSchedule _{\TimeIndex }} d\state$ is the normalization constant. We  initially target  the prior distribution  $\Target _{0}= p(\state  )$  which is generally easy to sample directly from  and then introduce the effect of the likelihood gradually in order to obtain at the end, $\TimeIndex =\NbIterSMC $, the complex posterior distribution of interest $\Target _{T}(\state  )=p(\state  | \Data )$ as target distribution. 
\item [b)] The idea is to  transform this problem in the standard SMC filtering framework, where the sequence of target distributions on the path-space, denoted by $\{ \tilde{\Target}_t \}_{t =1}^T$, which admits $\Target_t(x_t)$ as marginals, is defined on the product space, i.e., $\text{supp}(\tilde{\Target}_t) = \Theta \times \Theta \times ... \times \Theta = \Theta^t$. This novel sequence of joint target distributions $ \tilde{\Target }_{\TimeIndex }$ is defined as follows:
		 \begin{equation}\label{Eq_Artificial_TargetDistribution}
	 	 \tilde{\Target }_{\TimeIndex } (\state _{1:\TimeIndex })=\dfrac{\tilde{\UnNorTarget}_{\TimeIndex } (\state _{1:\TimeIndex }) }{\NormConst _{\TimeIndex }} ,
	 \end{equation}
	 where
	\begin{equation}\label{Eq_Artificial_UnNorTargetDistribution}
		 \tilde{\UnNorTarget  }_{\TimeIndex } (\state _{1:\TimeIndex })= \UnNorTarget  _{\TimeIndex }(\state_{\TimeIndex } ) \prod\limits _{k=1}^{\TimeIndex -1} \BackKernel _{k}(\state _{k+1},\state _{k}) ,
	\end{equation}
	in which  the artificial kernels introduced $\{\BackKernel _{k}\}_{k=1}^{\TimeIndex -1}$ are called \textit{backward} Markov kernels since $\BackKernel _{\TimeIndex } (\state _{\TimeIndex+1},\state _{\TimeIndex})$ denotes the  probability density of moving back from $\state _{\TimeIndex +1}$ to $\state _{\TimeIndex }$.  By using such a sequence of extended target distributions $\left \{  \tilde{\Target }_{\TimeIndex } \right\}_{\TimeIndex =1}^{\NbIterSMC } $ based on the introduction of backward kernels $\{\BackKernel _{k}\}_{k=1}^{\TimeIndex -1}$, sequential importance sampling can thus be utilized in the same manner as standard SMC filtering algorithms.
\end{enumerate}

	Within this framework, one may then work with the constructed sequence of distributions, $\widetilde{ \Target }_{\TimeIndex }$, under the standard SMC algorithm \cite{Doucet2001}. In summary, the SMC sampler algorithm therefore involves three stages:
\begin{enumerate}
\item{\underline{{\it Mutation:}}, where the particles are moved from $\State _{\TimeIndex -1}$ to $\State _{\TimeIndex }$ via a \textit{mutation kernel} $\Kernel _{\TimeIndex }(\state  _{\TimeIndex -1},\state_ {\TimeIndex })$ also called \textit{forward kernel};}
\item{\underline{{\it Correction:}},  where the particles are reweighted with respect to $\Target _{\TimeIndex }$ via the incremental importance weight (Equation (\ref{Eq_SMC_IncreWeights})); and}
\item{ \underline{{\it Selection:}}, where according to some measure of particle diversity, \Correct{such as} effective sample size, the weighted particles may be resampled in order to reduce the variability of the importance weights. }
\end{enumerate}

	In more detail, suppose that at time $\TimeIndex -1$,  we have a set of weighted particles $\left \{  \State_{1:\TimeIndex-1 }^{(\ParsIndexSMC )} ,\NormISWeight   _{\TimeIndex -1}^{(\ParsIndexSMC )} \right\} _{\ParsIndexSMC =1}^{\NbParticles }$ that approximates  $\tilde{\Target }_{\TimeIndex -1}$ via the empirical measure
	
	\begin{equation}\label{Eq_Approx_Artificial_TargetDistribution}
	{\tilde{\Target }}_{\TimeIndex-1}^N (d \state _{1:\TimeIndex -1})=\sum\limits _{\ParsIndexSMC =1}^{\NbParticles } \NormISWeight   _{\TimeIndex -1}^{(\ParsIndexSMC )} \delta_{\State _{1:\TimeIndex -1}^{(\ParsIndexSMC )}} (d \state _{1:\TimeIndex -1}).
	\end{equation}
	These particles are first propagated to the next distribution $\tilde{\Target }_{\TimeIndex }$ using \Correct{a} Markov kernel $\Kernel _{\TimeIndex } (\state_{\TimeIndex -1},\state _{\TimeIndex })$ to obtain the set of particles $\left  \{\State_{1:\TimeIndex }^{(\ParsIndexSMC )}  \right\} _{\ParsIndexSMC =1}^{\NbParticles }$. Importance Sampling (IS) is then used to correct for the discrepancy between the sampling distribution $\PropDist_{\TimeIndex }(\State_{1:\TimeIndex })$ defined \Correct{as}
	\begin{equation}\label{Eq_SMC_JointImportanceDist}
		\PropDist_{\TimeIndex } (\state_{1:\TimeIndex }^{(\ParsIndexSMC )})=\PropDist _{1}(\state  _{1}^{(\ParsIndexSMC )}) \prod_{k=2}^{\TimeIndex } \Kernel _{k}(\state  _{\TimeIndex -1}^{(\ParsIndexSMC )},\state  _{\TimeIndex }^{(\ParsIndexSMC )}) ,
	\end{equation}
	and $\tilde{\Target}_{\TimeIndex}(\State_{1:\TimeIndex })$. In this case the new expression for the unnormalized importance weights is given by
	\begin{equation}\label{Eq_SMC_ImportanceWeights}		
			\ISWeight_{\TimeIndex }^{(\ParsIndexSMC )} \varpropto \frac{ \tilde{\Target}_{\TimeIndex } (\state_{1:\TimeIndex }^{(\ParsIndexSMC )})}{\PropDist_{\TimeIndex } (\state _{1:\TimeIndex }^{(\ParsIndexSMC )} )} =\dfrac{\Target _{\TimeIndex }(\state  _{\TimeIndex }^{(\ParsIndexSMC )}) \prod_{s=1}^{\TimeIndex -1} \BackKernel_{s}(\state  _{s+1}^{(\ParsIndexSMC )},\state  _{s}^{(\ParsIndexSMC )})}{\PropDist _{1}(\state  _{1}^{(\ParsIndexSMC )}) \prod_{k=2}^{\TimeIndex } \Kernel _{k}(\state  _{k -1}^{(\ParsIndexSMC )},\state  _{k }^{(\ParsIndexSMC )})} \varpropto  \IncreWeight _{\TimeIndex } (\State _{\TimeIndex -1}^{(\ParsIndexSMC )},\State _{\TimeIndex }^{(\ParsIndexSMC )}) \ISWeight_{\TimeIndex -1}^{(\ParsIndexSMC )} ,
	\end{equation}
	where $\IncreWeight_{\TimeIndex }$, termed the (unnormalized) \textit{incremental weights}, are calculated as,
		\begin{equation}\label{Eq_SMC_IncreWeights}
		\IncreWeight_ {\TimeIndex } (\State _{\TimeIndex -1}^{(\ParsIndexSMC )},\State _{\TimeIndex }^{(\ParsIndexSMC )})=\dfrac{\UnNorTarget _{\TimeIndex }(\state  _{\TimeIndex }^{(\ParsIndexSMC )}) \BackKernel _{\TimeIndex -1}(\state  _{\TimeIndex }^{(\ParsIndexSMC )},\state  _{\TimeIndex -1}^{(\ParsIndexSMC )})}{\UnNorTarget _{\TimeIndex -1}(\state  _{\TimeIndex -1}^{(\ParsIndexSMC )}) \Kernel _{\TimeIndex }(\state  _{\TimeIndex -1}^{(\ParsIndexSMC )},\state  _{\TimeIndex }^{(\ParsIndexSMC )})} .
	\end{equation}
	\Correct{However, as in the particle filter, since the discrepancy between the target distribution $\tilde{\Target}_{\TimeIndex }$ and the proposal $\PropDist_{\TimeIndex }$ increases with $t$, the variance of the unnormalized importance weights tends therefore to increase as well, leading to a degeneracy of the particle approximation. A common criterion used in practice to check this problem is the effective sample size $\ESS$ which can be computed by:}
	\begin{equation}\label{Eq_SMC_ESS}
		\ESS_{\TimeIndex }=\left [  \sum\limits _{\ParsIndexSMC =1}^{\NbParticles } (\NormISWeight  _{\TimeIndex }^{(\ParsIndexSMC )})^{2} \right] ^{-1}=\dfrac{\left ( \sum\limits _{\ParsIndexSMC =1}^{\NbParticles } \ISWeight  _{\TimeIndex -1}^{(\ParsIndexSMC )} \UnNorIncreWeight _{\TimeIndex } (\State _{\TimeIndex -1}^{(\ParsIndexSMC )},\State _{\TimeIndex }^{(\ParsIndexSMC )}) \right) ^{2}}{\sum\limits _{j=1}^{\NbParticles }\left (   \ISWeight  _{\TimeIndex -1}^{(j)} \right) ^{2} \left (\UnNorIncreWeight _{\TimeIndex } (\State _{\TimeIndex -1}^{(j)},\State _{\TimeIndex }^{(j)})   \right) ^{2}} .
	\end{equation}
	If the degeneracy is too high, \Correct{i.e.,} the $\ESS_{\TimeIndex }$ is below a prespecified threshold, $\overline{\ESS}$, then a resampling step is performed. The particles with low weights are discarded whereas particles with high weights are duplicated. After resampling, the particles are equally weighted. 	
	
	Let us mention two interesting estimates from SMC samplers. Firstly, since  $\tilde{\Target }_{\TimeIndex }$ admits $\Target _{\TimeIndex }$ as marginals by construction, for any $1 \leq \TimeIndex \leq \NbIterSMC $ , \Correct{the} SMC sampler provides an estimate of this distribution 
	\begin{equation}\label{Eq_SMC_ApproxTarget}
		{\Target }_{\TimeIndex }^\NbParticles(d\state  )=\sum\limits _{\ParsIndexSMC =1}^{\NbParticles } \NormISWeight   _{\TimeIndex }^{(\ParsIndexSMC )} \delta_{\State _{\TimeIndex }^{(\ParsIndexSMC )}} (d\state  ),
	\end{equation}
	and an estimate of any expectations of some integrable function $\TestFunction (\cdot )$ with respect to this distribution given by
	\begin{equation}
		\Expec _{{\Target }_{\TimeIndex }^{\NbParticles }}\left[ \TestFunction (\state )\right]=\sum_{\ParsIndexSMC =1}^{\NbParticles } \NormISWeight _{\TimeIndex }^{(\ParsIndexSMC )} \TestFunction (\state _{\TimeIndex }^{(\ParsIndexSMC )}) .
		\label{ExpectationApproxSMCSampler}
	\end{equation}
	Secondly,  the \Correct{estimated ratio} of normalizing constants $\dfrac{\NormConst _{\TimeIndex }}{\NormConst _{\TimeIndex -1}}=\dfrac{\int \UnNorTarget _{\TimeIndex }(\state ) d \state  }{\int \UnNorTarget _{\TimeIndex-1 }(\state ) d \state  }$ is given by
	\begin{equation}\label{Eq_Approx_Ratio_NormalizingConstant}		
		\widehat{\dfrac{\NormConst_{\TimeIndex }}{\NormConst_{\TimeIndex -1}}}=\sum\limits _{\ParsIndexSMC =1}^{\NbParticles } \NormISWeight  _{\TimeIndex -1}^{(\ParsIndexSMC )} \UnNorIncreWeight _{\TimeIndex } (\State _{\TimeIndex -1}^{(\ParsIndexSMC )},\State _{\TimeIndex }^{(\ParsIndexSMC )}) .
	\end{equation}
	Consequently, the estimate of $\dfrac{\NormConst_{\TimeIndex } }{\NormConst_{1 }}$ \Correct{is}
	\begin{equation}\label{Eq_Approx_Ratio_NormalizingConstantZ0Zt}	
		\widehat{\dfrac{\NormConst_{\TimeIndex }}{\NormConst_{1 }}}=\prod\limits _{k=2}^{\TimeIndex} \widehat{\dfrac{\NormConst_{k }}{\NormConst_{k-1 }}}=\prod\limits _{k=2}^{\TimeIndex} \sum\limits _{\ParsIndexSMC =1}^{\NbParticles }  \NormISWeight  _{k -1}^{(\ParsIndexSMC )} \UnNorIncreWeight _{k } (\State _{k -1}^{(\ParsIndexSMC )},\State _{k }^{(\ParsIndexSMC )}) .
	\end{equation}
	If the resampling scheme used is unbiased, then (\ref{Eq_Approx_Ratio_NormalizingConstantZ0Zt}) is also unbiased  whatever the number of particles used \cite{DelMoral2000}. Moreover, the complexity of this algorithm is  ${\cal O}(\NbParticles )$ per time step and it can be easily  parallelized.  
	
	Finally, let us note that there exists a few other SMC methods appropriate for static inference such as annealed importance sampling \cite{Neal2001}, the sequential particle filter of \cite{Chopin2002} and population Monte Carlo \cite{Cappe2004} but all of these methods can be regarded as  a special case of the SMC sampler framework.
	
	\subsection{On the choice of the sequence of target distributions and mutation/backward kernels}\label{AlgorithmSettingsChap1}

	The algorithm presented in the previous subsection is very general. There is a wide range of possible choices to consider when designing an SMC sampler algorithm, the appropriate sequence of distributions $\{\Target_{\TimeIndex  }\}_{1 \leq \TimeIndex \leq \NbIterSMC }$, the choice of both the mutation kernel $\{\Kernel _{\TimeIndex }\}_{2 \leq \TimeIndex \leq \NbIterSMC }$ and the backward mutation kernel $\{\BackKernel _{\TimeIndex-1 }\}_{\TimeIndex =2}^{\NbIterSMC }$ (for a given mutation kernels), see details in \cite{Peters:2005ti,DelMoral2006,peters2012sequential}. In this subsection, we provide a discussion on how to choose these parameters of the algorithm in practice.
	
\subsubsection{\underline{Sequence of distributions $\Target _{\TimeIndex }$}} 

There are many potential \Correct{choices} for $\{\Target_{\TimeIndex  }  \}$  leading to various integration and optimization algorithms. As  a special case, we can set $\Target_{\TimeIndex  }=\MainTarget   $ for all $\TimeIndex \in \TimeSet $. Alternatively, to maximize $\MainTarget  (\state  )$, we could consider $\Target_{\TimeIndex  }(\state_ {\TimeIndex} )=[\MainTarget (\state_ {\TimeIndex} )]^{\xi_{\TimeIndex}}$ for an increasing schedule $\{\xi_{\TimeIndex}\}_{\TimeIndex \in \TimeSet}$ to ensure $\Target _{\NbIterSMC }(\state )$ is concentrated around the set of global maxima of $\MainTarget  (\state  )$. In the context of Bayesian inference for static parameters which is the main focus of this paper,  one can consider $\Target_{\TimeIndex  }(\state  )=p(\state |\OneData_{1},\cdots, \OneData_{\TimeIndex})$, which corresponds to \textit{data tempered} schedule. 

In this paper, we are interested in the \textit{likelihood tempered} target sequence, that has been proposed in \cite{Neal2001},
\begin{equation}\label{SequenceSMCsampler}
\Target _{\TimeIndex }(\state  )=\frac{\UnNorTarget_t(\state)}{\NormConst_t}\varpropto p(\state  ) p (\Data | \state  )^{\CoolSchedule _{\TimeIndex }} ,
\end{equation} 
where $\left \{   \CoolSchedule _{\TimeIndex }\right\} $ is a non-decreasing temperature schedule with $\CoolSchedule _{0}=0$ and $\CoolSchedule _{\NbIterSMC }=1$. We thus sample initially from the prior distribution  $\Target _{0}= p(\state  )$ directly and introduce the effect of the likelihood gradually in order to obtain at the end $\TimeIndex =\NbIterSMC $ an approximation of the posterior distribution $p(\state  | \Data )$. As discussed in \cite{Neal2001}, tempering the likelihood could significantly  improve the exploration of the state space in complex multimodal posterior distribution. From Eq. (\ref{Eq_Approx_Ratio_NormalizingConstantZ0Zt}), the normalizing constant of the posterior target distribution which corresponds to the \textit{marginal likelihood}, $p (\Data)$, can be approximated with SMC samplers as:
	\begin{equation}\label{Eq_Approx_BayesianEvidence}
		\NormConst _{\NbIterSMC }=\NormConst _{1} \prod\limits _{\TimeIndex =2}^{\NbIterSMC }\dfrac{\NormConst _{\TimeIndex }}{\NormConst _{\TimeIndex -1}} \approx \prod\limits _{\TimeIndex =2}^{\NbIterSMC } \sum\limits _{\ParsIndexSMC =1}^{\NbParticles }  \NormISWeight  _{\TimeIndex  -1}^{(\ParsIndexSMC )} \UnNorIncreWeight _{\TimeIndex  } (\State _{\TimeIndex  -1}^{(\ParsIndexSMC )},\State _{\TimeIndex  }^{(\ParsIndexSMC )})
	\end{equation}
	where $Z_{\TimeIndex }=\int p (\Data | \state )^{\CoolSchedule _{\TimeIndex }}p (\state ) d \state  $  corresponds to the normalizing constant of the target distribution at iteration $\TimeIndex $ (thus $\NormConst_1=\int p (\state ) d \state=1$). The approximation of an expectation with respect to the posterior is given by:
	\begin{equation}\label{FinalExpectationApproxSMC}
	\Expec_{\Target^N} \left[\TestFunction(\state) \right] =\sum_{i=1}^{\NbParticles } \NormISWeight_{\NbIterSMC }^{(i)} \TestFunction (\State_{\NbIterSMC }^{(i)}).
	\end{equation}

\subsubsection{\underline{Sequence of mutation kernels $\Kernel _{\TimeIndex }$}}
	
		The performance of SMC samplers \Correct{depends} heavily upon the selection of the transition kernels $\left \{  \Kernel _{\TimeIndex } \right\} _{\TimeIndex =2}^{\NbIterSMC }$  and the auxiliary backward kernels $\left \{ \BackKernel _{\TimeIndex -1}  \right\}_{\TimeIndex =2}^{\NbIterSMC } $. 
		  There are many possible choices for $\Kernel _{\TimeIndex }$ which have been discussed in \cite{Peters:2005ti,DelMoral2006,peters2012sequential}. In this study, we propose to employ  MCMC kernels of invariant distribution $\Target _{\TimeIndex }$ for $\Kernel _{\TimeIndex }$. This is an attractive strategy since we can use the vast literature on the design of efficient MCMC algorithms to build a good importance distributions (See \cite{Robert2004}). 
		  
		  More precisely, since we are interested in complex models with potentially high-dimensional and multimodal posterior distribution, a series of Metropolis-within-Gibbs kernels allowing local moves will be employed in order to successively move the $B$ sub-blocks of the state of interest, $\State =[\Block _{1}, \Block _{2}, \cdots, \Block _{B}]$. A random walk proposal distribution is used for each sub-block with a multivariate Gaussian distribution as proposal:
\begin{equation}
\Block _{b,t}^*=\Block _{b,t-1} + {\bm \epsilon}_{b,t},
\end{equation}		  
\Correct{in which $\epsilon_{b,t}$ is a Gaussian random variable with zero mean and covariance matrix   ${\bm \Sigma}_{b,t}$}.  As with any sampling \Correct{algorithm}, faster mixing does not harm performance and in some cases will considerably improve it. In the particular case of Metropolis-Hastings kernels, the mixing speed relies on adequate proposal scales. As a consequence, we adopt the strategy proposed in \cite{jasra2011}. The authors applied an idea used within adaptive MCMC methods \cite{andrieu2006} to SMC samplers by using the variance of the parameters estimated from its particle system approximation as the proposal scale for the next iteration, \Correct{i.e.,} the covariance matrix of the random-walk move for the $b$-th sub-block at time $t$ is given by:
\begin{eqnarray}\label{ComputationAdaptiveScalingAWG}
{\bm \Sigma}_{b,t}&=&\sum_{m=1}^\NbParticles \NormISWeight_{\TimeIndex-1 }^{(m)} \left(\Block _{b,\TimeIndex-1}^{(m)} - {\bm \mu}_{b,\TimeIndex-1} \right) \left(\Block _{b,\TimeIndex-1}^{(m)} - {\bm \mu}_{b,\TimeIndex-1} \right)^T ,\\
\text{with} && {\bm \mu}_{b,\TimeIndex-1}=\sum_{m=1}^\NbParticles \NormISWeight_{\TimeIndex-1 }^{(m)} \Block _{b,\TimeIndex-1}^{(m)}. \nonumber
\end{eqnarray} 
  The motivation is that if $\Target _{\TimeIndex-1 }$ is close to $\Target _{\TimeIndex}$ (which is recommended for having an efficient SMC algorithm), then the variance estimated at iteration $t-1$ will provide a sensible scaling at time $t$. \Correct{This adaptive Metropolis within Gibbs used in the implementation of the SMC sampler through this  paper is summarized in Algorithm \ref{Algo_AMWG}.}
 	  
 	  In difficult problems, other approaches could be added in order to have appropriate scaling adaptation; one approach demonstrated in \cite{jasra2011} is to simply employ a pair of acceptance rate thresholds and to alter the proposal scale given by Eq. (\ref{ComputationAdaptiveScalingAWG}) whenever the acceptance rate falls outside those threshold values. This scheme is to ensure that the acceptance rates in the Metropolis-Hastings steps do not get too large or small. Through all this paper, we use this procedure which consists for example to multiply the covariance matrix by 5 (resp. 1/5)  if the rate exceeded 0.7 (resp. fell below 0.2).

		\subsubsection{\underline{Sequence of  backward kernels $\BackKernel _{\TimeIndex }$}}

		The backward kernel $\BackKernel _{\TimeIndex }$ is arbitrary, however as discussed in \cite{Peters:2005ti,DelMoral2006,peters2012sequential}, it should be optimized with respect to mutation kernel $\Kernel _{\TimeIndex }$ to obtain good performance.  \cite{Peters:2005ti,DelMoral2006,peters2012sequential} establish that the backward kernel which minimize the variance of the unnormalized importance weights, $\ISWeight _{\TimeIndex }$, are given by
	  \begin{equation}\label{Eq_Optimal_BackwardKernel}
	  	\BackKernel _{\TimeIndex -1}^{\text{opt}} (\state  _{\TimeIndex },\state  _{\TimeIndex -1 })=\dfrac{\PropDist _{\TimeIndex -1}(\state  _{\TimeIndex -1}) \Kernel _{\TimeIndex }(\state  _{\TimeIndex -1},\state  _{\TimeIndex })}{\PropDist _{\TimeIndex }(\state  _{\TimeIndex })} .
	  \end{equation}
	  However, as discussed in Section \ref{MethodologySMCSamplers}, it is typically impossible to use these optimal kernels as they rely on marginal distributions of the joint proposal distribution defined in Eq. (\ref{Eq_SMC_JointImportanceDist})  which do not admit any closed form expression, especially if an MCMC kernel is used as $\Kernel _{\TimeIndex }$ which is  $\Target _{\TimeIndex }$-invariant distribution. Thus we can either choose to approximate $\BackKernel _{\TimeIndex }^{\text{opt}}$ or choose kernels $\BackKernel _{\TimeIndex }$ so that the importance weights are easily calculated or have a familiar form. As discussed in  \cite{Peters:2005ti,DelMoral2006}, if an MCMC kernel is used as forward kernel, the following $\BackKernel _{\TimeIndex }$ is employed 
	\begin{equation}\label{Eq_SubOptimal_BackwardKernel}		
		\BackKernel _{\TimeIndex -1} (\state _{\TimeIndex },\state _{\TimeIndex -1})=\dfrac{\Target _{\TimeIndex }(\state _{\TimeIndex -1}) \Kernel _{\TimeIndex }(\state _{\TimeIndex -1},\state _{\TimeIndex })}{\Target _{\TimeIndex }(\state _{\TimeIndex})},
	\end{equation}
which is a good approximation of the optimal backward if the discrepancy between $\Target _{\TimeIndex }$ and $\Target _{\TimeIndex -1}$ is small; note that (\ref{Eq_SubOptimal_BackwardKernel}) is the reversal Markov kernel associated with $\Kernel_{\TimeIndex }$. In this case,  the unnormalized incremental weights becomes	
	\begin{equation}\label{Eq_UnNormalized_IncrementWeights1}
		\UnNorIncreWeight _{\TimeIndex }^{(\ParsIndexSMC )} (\State _{\TimeIndex -1}^{(\ParsIndexSMC )},\State _{\TimeIndex }^{(\ParsIndexSMC )})=\dfrac{\UnNorTarget_{\TimeIndex }(\state  _{\TimeIndex -1}^{(\ParsIndexSMC )})}{\UnNorTarget_{\TimeIndex-1 }(\state  _{\TimeIndex -1}^{(\ParsIndexSMC )})}=  p(\Data | \state _{\TimeIndex -1}^{(\ParsIndexSMC )}  )^{(\CoolSchedule _{\TimeIndex }-\CoolSchedule _{\TimeIndex -1})} 
	\end{equation}	
	This expression (\ref{Eq_UnNormalized_IncrementWeights1}) is remarkably easy to compute and valid regardless of the MCMC \Correct{kernel} adopted. Note that $\CoolSchedule _{\TimeIndex }-\CoolSchedule _{\TimeIndex -1}$ is the step length of the cooling schedule of the likelihood at time $\TimeIndex $. As we choose this step larger, the discrepancy between $\Target _{\TimeIndex }$ and $\Target _{\TimeIndex -1}$ increases, leading to an increase as the variance of the importance approximation. Thus, it is important to construct a smooth sequence of distributions $\left \{ \Target _{\TimeIndex }  \right\} _{0 \leq \TimeIndex \leq \NbIterSMC }$ by judicious choice of an associated real sequence $\left \{ \CoolSchedule _{\TimeIndex }  \right\} _{\TimeIndex =0}^{\NbIterSMC }$.
	
	Let us remark that when such backward kernel is used, the unnormalized incremental weights in Eq. (\ref{Eq_UnNormalized_IncrementWeights1}) at time $t$ does not depend on the particle value at time $t$ but just on the previous particle set. In such a case, the particles $\left\{ \state _{\TimeIndex }^{(\ParsIndexSMC )}   \right\}$ should be sampled after the weights 	$\left\{\ISWeight_{\TimeIndex }^{(\ParsIndexSMC )}\right\}$ have been computed and after the particle approximation  $\left\{ \ISWeight_{\TimeIndex }^{(\ParsIndexSMC )},\state _{\TimeIndex-1 }^{(\ParsIndexSMC )}   \right\}$ has possibly been resampled. 
	
	Based on these discussions regarding the different possible choices, the SMC sampler that will be used for Bayesian inference through this paper is summarized in Algorithm \ref{Algo_SMCSpecific}.

 \begin{algorithm}[h]  
			\caption{SMC Sampler Algorithm}
			\label{Algo_SMCSpecific}
			 \begin{algorithmic}[1]
 			 	 \small 
  			 	 \STATE \underline{Initialize particle system} 
   				 \STATE Sample $\left \{\State _{1}^{(\ParsIndexSMC )}   \right\}_{\ParsIndexSMC =1}^{\NbParticles } \sim \PropDist _{1} (\cdot)$ and compute  $\NormISWeight  _{1}^{(\ParsIndexSMC )}=\left ( \frac{\UnNorTarget _{1}(\state _{1}^{(\ParsIndexSMC )})}{\PropDist _{1}(\state _{1}^{(\ParsIndexSMC )})}  \right) \left [  \sum_{j =1}^{\NbParticles }\frac{\UnNorTarget _{1}(\state _{1}^{(j)})}{\PropDist _{1}(\state _{1}^{(j )})} \right] ^{-1} $ and do resampling if $\ESS < \overline{\ESS}$
  				 \FOR {$\TimeIndex =2, \ldots, \NbIterSMC $}
					\STATE \underline{Computation of the weights:} for each $\ParsIndexSMC =1,\ldots,\NbParticles $ 
  							\vspace{0.1cm} $$\ISWeight _{\TimeIndex }^{(\ParsIndexSMC )}= \NormISWeight     _{\TimeIndex -1}^{(\ParsIndexSMC )} p(\Data | \state _{\TimeIndex -1}^{(\ParsIndexSMC )}  )^{(\CoolSchedule _{\TimeIndex }-\CoolSchedule _{\TimeIndex -1})}$$\vspace{0.1cm}    
   						Normalization of the weights : $\NormISWeight _{\TimeIndex }^{(\ParsIndexSMC )}=\ISWeight _{\TimeIndex }^{(\ParsIndexSMC )}\left [  \sum_{j=1} ^{\NbParticles } \ISWeight _{\TimeIndex }^{(j)}\right] ^{-1}$  				 
  				 		 \STATE \underline{Selection:} if $ESS<\overline{\ESS}$ then Resample
  				 		 \STATE \underline{Mutation:} for each $\ParsIndexSMC =1,\ldots,\NbParticles $ : Sample $\State _{\TimeIndex }^{\ParsIndexSMC } \sim \Kernel _{\TimeIndex }(\state_{\TimeIndex -1}^{(\ParsIndexSMC )};\cdot)$ where $\Kernel _{\TimeIndex }(\cdot;\cdot)$ is a $\Target _{\TimeIndex }(\cdot)$ invariant Markov kernel described in more details in Algo. \ref{Algo_AMWG}.
  				 \ENDFOR
			\end{algorithmic}
	\end{algorithm}

	 \begin{algorithm}[h]  
			\caption{Adaptive Metropolis-within-Gibbs Kernel $\Kernel _{\TimeIndex }(\cdot;\cdot)$ for the $m$-th particle}
			\label{Algo_AMWG}
			 \begin{algorithmic}[1]
 			 	 \small 
  				\STATE \underline{Initialization} Set $\State^0=[\Block _{1}^{0}, \ldots, \Block _{B}^{0}]=\State_{t-1}^{(m)} =[\Block _{1,t-1}^{(m)}, \ldots, \Block _{B,t-1}^{(m)}]$
  				 \FOR {$i =1, \ldots, \NbMCMCmove $}
  				 \FOR{$b=1,\ldots,B$}
					\STATE Sample $\Block_b^* \sim {\cal N} \left( \Block_b^{i-1}, {\bm \Sigma}_{b,t} \right)$ with ${\bm \Sigma}_{b,t}$ defined in Eq. \ref{ComputationAdaptiveScalingAWG}
					\STATE Compute the Acceptance ratio:
					$$\AcceptRatio (\Block_b ^{*},\Block_b ^{i -1})=\min \left \{ 1,\dfrac{p(\Data | \state^{*}  )^{\CoolSchedule _{\TimeIndex }} p(\state^{*} )}{p(\Data | \state^{i-1}  )^{\CoolSchedule _{\TimeIndex }} p(\state^{i-1} )}\right\} $$
					with $\state^{*} =[\Block _{1}^{i}, \ldots,\Block _{b-1}^{i}, \Block _{b}^{*},\Block _{b+1}^{i-1},\ldots,\Block _{B,t-1}^{0}]$
					and $\state^{i-1} =[\Block _{1}^{i}, \ldots,\Block _{b-1}^{i}, \Block _{b}^{i-1},\Block _{b+1}^{i-1},\ldots,\Block _{B,t-1}^{0}]$
			\STATE Sample random variate $u$ from $\Uniform (0,1)$   					
					\IF{$u \leq \AcceptRatio (\State ^{*},\State ^{i-1})$}
						\STATE $\Block_b ^{i}=\Block_b ^{*}$	
					\ELSE
						\STATE $\Block_b ^{i}=\Block_b ^{i-1}$
					\ENDIF			
					\ENDFOR		
  				 \ENDFOR
  				\STATE  Set the new particle value at time $t$ as $\State_{t}^{(m)} =[\Block _{1}^{\NbMCMCmove}, \ldots, \Block _{B}^{\NbMCMCmove}]$
			\end{algorithmic}
	\end{algorithm}

\subsection{Asymptotic analysis of the SMC sampler-based estimator} \label{SMC_convergence}

		In the remainder of this section,  we are interested in the convergence results of SMC Samplers. We aim  at analyzing the influence of the choice of the sequence of intermediate target distributions in the performance of the SMC sampler  when resampling is performed before the sampling step. 
As discussed in the previous section, one of the main attractive \Correct{properties} of the SMC sampler is to be able to use some local moves (using an MCMC kernel) in order to draw the particles at the next iteration. Such local moves are particularly interesting when the state of interest is high-dimensional. As discussed in Section \ref{MethodologySMCSamplers}, when such an MCMC kernel is used as a forward kernel in the SMC sampler, $\Kernel_{\TimeIndex}$, the backward kernel used in order to be able to compute the incremental weight is given in Eq. (\ref{Eq_SubOptimal_BackwardKernel}). 
Moreover, in order to obtain convergence results that are easy to analyze and utilize, we assume that the MCMC kernel used is perfectly mixing, \Correct{i.e.} 
$\Kernel_{\TimeIndex}(\state_{\TimeIndex-1},\state_{\TimeIndex})=\pi_{\TimeIndex}(\state_{\TimeIndex}) $.

Under these two assumptions, we derive the asymptotic variance of the SMC sampler estimates when resampling is performed before the sampling stage at each iteration.
	\begin{proposition} \label{Proposition_Convergence_RB}
	\begin{itshape}
Under perfect mixing assumption and if the backward kernel given in Eq. (\ref{Eq_SubOptimal_BackwardKernel}) is used, we obtain the following results:
\begin{enumerate}

\item For \Correct{the  expectation estimator:}	\begin{equation}
							N^{\frac{1}{2}}\left \{\Expec_{\pi^{N}_{\TimeIndex}} (\varphi)-\Expec_{\pi_{\TimeIndex}}(\varphi)\right \}\Rightarrow \mathcal{N} (0,\sigma_{SMC,\TimeIndex}^{2}(\varphi))
						\end{equation}	
						with
							\begin{equation}
								\sigma_{SMC,\TimeIndex}^{2}(\varphi)=\left \{\Expec_{\pi_{\TimeIndex}}( \varphi^{2}(\state))-\Expec^{2}_{\pi_{\TimeIndex}}(\varphi(\state))\right \}=\Var_{\pi_{\TimeIndex}}(\varphi(\state))
							\end{equation}
\item For \Correct{the normalizing constant estimator:}
\begin{equation}
							N^{\frac{1}{2}} \left \{ \log\left (\widehat{\dfrac{\NormConst_{\TimeIndex}}{\NormConst_{1}}}\right )-\log \left (\dfrac{\NormConst_{\TimeIndex}}{\NormConst_{1}}\right ) \right \} \Rightarrow \mathcal{N} (0,\sigma^{2}_{SMC,\TimeIndex})
						\end{equation}
						
						with 
						\begin{equation}
							 	\begin{aligned}
							 		\sigma^{2}_{SMC,\TimeIndex}&=\int \dfrac{\pi_{2}^{2}(\state_{1})}{\eta_{1}(\state_{1})}d\state_{1}&+\sum\limits_{k=2}^{\TimeIndex-1} \int \dfrac{\pi_{k+1}^{2}(\state_{k})}{\pi_{k}(\state_{k})}d\state_{k} -(\TimeIndex-1)
							 	\end{aligned}
							 	\label{OptimalVar_ResamplingBefore}
						\end{equation}

\end{enumerate}
\end{itshape}
\end{proposition}
\begin{proof} 
See Appendix.
\end{proof}

\noindent\textbf{Remark 1:} As expected, we can conclude from these results that even if a perfect mixing MCMC kernel is used, the variance of the \Correct{estimator associated with} the normalizing constant in Eq. (\ref{OptimalVar_ResamplingBefore}) still depends on all the sequence of target distributions as a cumulative sum of the discrepancy between two consecutive target distributions. 

In the next section, we will use this result in order to design an automatic procedure for the selection of the sequence of target distributions and more especially the evolution of the cooling schedule that defines completely this sequence in Eq. (\ref{SequenceSMCsampler}).

	\section{Adaptive Sequence of Target Distributions}\label{Sec_AdaptiveTarget}	

\subsection{Existing approaches}

Several statistical approaches have been proposed in order to automatically obtain such a schedule via the optimization of some criteria, which are known as \textit{on-line} schemes.  \cite{jasra2011} 
		 proposed an adaptive selection method based on controlling the rate of the effective sample size ($\ESS_{\TimeIndex }$), defined in  (\ref{Eq_SMC_ESS}). 
\Correct{This scheme thus provides an automatic method to obtain the tempering schedule such that the $\ESS$ decays in a regular predefined way.  However, one major drawback of such an approach is that the $\ESS_{\TimeIndex }$ of the current sample weights corresponds to some empirical measure of the accumulated discrepancy between the proposal and the target distribution since the last resampling time. As a consequence, it does not really represent the dissimilarity between each pair of successive distributions   unless resampling is conducted after every iteration.
		 
		 In order to handle this problem,  \cite{Yan2013} proposed a slight modification of the $\ESS$, named the \textit{conditional} $\ESS$ ($\CESS$), by considering how good an importance sampling proposal  $\Target _{k,\TimeIndex-1 } $ would be for the estimation of expectation under $\Target _{\TimeIndex } $.  At the $t$-th iteration, this quantity  is defined as follows: }	
	
		\begin{equation}\label{Eq_CESS}
			\CESS_{\TimeIndex } =\left [\sum\limits_{i=1}^{\NbParticles }\NbParticles \NormISWeight_{\TimeIndex-1 } ^{(i)} \left (\dfrac{\UnNorIncreWeight _{\TimeIndex } ^{(i)}}{\sum_{j = 1}^{\NbParticles }\NbParticles \NormISWeight_{\TimeIndex-1 } ^{(j)}\UnNorIncreWeight _{\TimeIndex } ^{(j)}}\right )^{2}\right ]^{-1}=\dfrac{\left (\sum_{i=1}^{\NbParticles }\NormISWeight_{\TimeIndex-1 } ^{(i)}\UnNorIncreWeight _{\TimeIndex } ^{(i)}\right )^{2}}{\sum_{j= 1}^{\NbParticles }\frac{1}{\NbParticles }\NormISWeight_{\TimeIndex-1 } ^{(j)}(\UnNorIncreWeight _{\TimeIndex } ^{(j)})^{2}} .
		\end{equation}

		Nevertheless, by using either $\ESS$ or $\CESS$ criterion, the number of steps $\NbIterSMC $ of the SMC samplers completely depends on the complexity of the integration problem at hand and is not known in advance.  In other words, for  either fixed $\ESS^{\star}$ or fixed $\CESS^{\star}$, the associated sequence $\{\CoolSchedule _{\TimeIndex } \}_{\TimeIndex = 1}^{\NbIterSMC }$ is an on-line self-tuning parameter. Smaller values significantly speed up the Sequential Monte Carlo algorithm but lead to a higher variation in the results. Consequently, we are not able to control the total complexity of the algorithm, and it is typically impossible to obtain the comprehensive view of the behavior of the cooling schedule $\left \{  \CoolSchedule_{\TimeIndex }  \right\}$ a priori, instead one has to wait until the algorithm is completed. 

\subsection{Proposed adaptive cooling strategy}
\label{ProposedAdaptiveScheme}

	In this paper, we propose an alternative strategy to choose  the sequence of target distributions adaptively to the problem under study. In particular, we propose to \Correct{consider} the sequence of distributions which minimizes the variance of the particle approximation of the normalizing constant derived previously in Eq. (\ref{OptimalVar_ResamplingBefore}). This strategy is thus based on a global optimization of cooling schedule $\left \{  \CoolSchedule _{\TimeIndex } \right\} $ which enable us to control the complexity of the algorithm by determining before any simulation the number of SMC iterations $\NbIterSMC $. In this way we obtain what will be referred to as an \textit{off-line} scheme, and we will obtain the complete view of the cooling schedule performance before running the SMC sampler.
		
\subsubsection{\underline{Objective function and Optimization procedure}}
		
		By carrying out our criterion, we have to find $\NbIterSMC-1$ positive step lengths ${\bm \IncreStep}= \{\IncreStep _{\TimeIndex }\}_{\TimeIndex =2}^{\NbIterSMC }$ , defined as $\CoolSchedule  _{\TimeIndex }-\CoolSchedule  _{t-1}$ such that $\sum_{t=2}^T \IncreStep _{\TimeIndex } =1$, which minimize the asymptotic variance given in Eq. (\ref{OptimalVar_ResamplingBefore}). 
		\Correct{Here, we are aiming at finding
		\begin{eqnarray}
		\label{Eq_OptGamma}
		\widehat{\bm \IncreStep}=\left \{\widehat{\IncreStep} _{2 }, \ldots,\widehat{\IncreStep}  _{T } \right \}&=&\underset{\left \{\IncreStep  _{2 },  \ldots,\IncreStep  _{T }\right \}  }{\argmin}  \quad \sum\limits_{\TimeIndex =1}^{\NbIterSMC -1} \int \dfrac{\Target _{\TimeIndex +1}^{2}(\State _{\TimeIndex })}{\Target _{\TimeIndex }(\State _{\TimeIndex })}d\State _{\TimeIndex }-(\NbIterSMC -1) \\
		&&\text{subject to } \sum_{t=2}^T \IncreStep _{\TimeIndex } =1 \text{ and } \forall m=2,\ldots,T: \IncreStep_m\geq0 \nonumber
	\end{eqnarray}
	where
	\begin{equation}\label{Eq_NormalizedTargetDistribution}
		\Target _{\TimeIndex }(\theta)= \dfrac{p(y|\theta)^{\CoolSchedule  _{\TimeIndex }}p(\theta)}{\int p(y|\theta)^{\CoolSchedule  _{\TimeIndex }}p(\theta) d\theta}=\dfrac{p(y|\theta)^{\CoolSchedule  _{\TimeIndex }}p(\theta)}{Z_{\TimeIndex }} \text{ with } \CoolSchedule_t=\sum_{m=2}^t \IncreStep_m .
	\end{equation}}

	Equation (\ref{Eq_OptGamma}) involves $\NbIterSMC -1$ integrals and each integral  represents, as discussed in Section \ref{SMC_convergence}, a dissimilarity measure between each pair of successive distributions. The main difficulty in carrying out this construction is that these integrals are generally intractable, typically requiring approximation.

In order to avoid the use of numerical methods to approximate the $\NbIterSMC-1$ integrals which could be very challenging to do if $\state$ is high-dimensional, we propose instead to approximate each target distribution $\Target _{\TimeIndex }(\State )$ by a multivariate normal distribution. Indeed, from the connection between these integrals and the R\'enyi divergence between two distributions, an analytical expression for the asymptotic variance to minimize can be obtained by using the following result \cite{Gil2013124}:

\vspace*{0.2cm}
\noindent\textit{For Gaussian multivariate distribution $f_{1}=\mathcal{N}(\mu_{1},\Sigma_{1})$ and  $f_{2}=\mathcal{N}(\mu_{2},\Sigma_{2})$ we have}
		\begin{equation}
			\begin{aligned}
				\int f^{\alpha}_{1} (x) f_{2}^{1-\alpha} (x) dx&= \dfrac{\det \left (\alpha\Sigma_{2}+(1-\alpha)\Sigma_{1}\right )^{-\frac{1}{2}}}{\det (\Sigma_{1})^{\frac{\alpha-1}{2}}\det(\Sigma_{2})^{-\frac{\alpha}{2}}} \\
				&\times \exp \left \{ \dfrac{\alpha (\alpha-1)}{2} (\mu_{1}-\mu_{2})^{T}(\alpha \Sigma_{2}+(1-\alpha)\Sigma_{1})^{-1}(\mu_{1}-\mu_{2})\right \}
			\end{aligned}
		\end{equation}
\noindent\textit{which is finite iff }$\alpha\Sigma_{1}^{-1}+(1-\alpha)\Sigma_{2}^{-1}$\textit{ is positive definite.}
\vspace*{0.2cm}

Finally, a nonlinear optimization technique, such as for example the Nelder-Mead algorithm \cite{NelderMead65}, can be used to solve this optimization in order to obtain the value $\widehat{\bm \IncreStep}$. 

\subsubsection{\underline{Normal approximations of intermediate target distributions}}

In order to find the value $\widehat{\bm \IncreStep}$ that minimizes the asymptotic variance of the estimate of the normalizing constant, we need to approximate the $\NbIterSMC$ intermediate target distributions, $\pi_\TimeIndex$ for $\TimeIndex=1,\cdots,\NbIterSMC$ by \Correct{ multivariate normal distributions, i.e.,:
\begin{equation}
\pi_{\TimeIndex}(\state)  \propto p(\Data|\State)^{\CoolSchedule_{\TimeIndex}} p(\State) \approx \mathcal{N}(\state|{\bm \mu}_{\TimeIndex},{\bm \Sigma}_{\TimeIndex}) .
\end{equation}
}
 In order to reduce the complexity associated with these $\NbIterSMC$ different normal\Correct{ approximations} of the intermediate target distributions \Correct{(which consists in finding both $\NbIterSMC$ mean vectors $\left\{{\bm \mu}_{\TimeIndex}\right\}_{\TimeIndex=1}^\NbIterSMC$ and  covariance matrices $\left\{{\bm \Sigma}_{\TimeIndex}\right\}_{\TimeIndex=1}^\NbIterSMC$)}, we propose to only approximate the prior and the posterior \Correct{$ \PostDist (\State|\Data)$} distribution and thus deduce all normal approximation required by using the convenient properties of the normal distribution.
 
 Indeed, approximating both the prior and the posterior by normal distributions with parameters $({\bm \mu}_p, {\bm \Sigma}_p)$ and $({\bm \mu}_{T}, {\bm \Sigma}_{T})$ respectively, leads to a normal likelihood  approximation with
		\begin{align}
		\begin{split}
			{\bm \Sigma}_l & = \left({\bm \Sigma}_{T}^{-1} -  {\bm \Sigma}_p^{-1}\right)^{-1} , \\
			{\bm \mu}_l & = {\bm \Sigma}_l  \left( {\bm \Sigma}_{T}^{-1}  {\bm \mu}_T -{\bm \Sigma}_{p}^{-1}  {\bm \mu}_p  \right) . \label{paralikelihood}
		\end{split}
	\end{align}
	Moreover, since a tempered normal is proportional to a normal with only a modification of the covariance and  also the product of 2 multivariate normals is a multivariate normal distribution, the $\TimeIndex $-th target distribution can therefore  be approximated by :
\begin{equation}
		\Target _{\TimeIndex }(\State)\approx \Normal(\State| {\bm \mu}_{\TimeIndex }, {\bm \Sigma}_{\TimeIndex }),
		\label{TargetApprox}
	\end{equation}
	with 
	\begin{align}
		\begin{split}
			{\bm \Sigma}_{\TimeIndex } & = \left( {\bm \Sigma}_p^{-1} +\CoolSchedule _{\TimeIndex } {\bm \Sigma}_l^{-1}  \right)^{-1} , \\
			{\bm \mu}_{\TimeIndex } & = {\bm \Sigma}_{\TimeIndex }  \left( {\bm \Sigma}_p^{-1}  {\bm \mu}_p +\CoolSchedule _{\TimeIndex } {\bm \Sigma}_l^{-1}  {\bm \mu}_l  \right). \label{ParaTargetApprox}
		\end{split}
	\end{align}
	
Only the prior and the posterior require normal approximations which can be performed using either  Laplace's method \cite{MacKay:2007fk} (which requires to be able to compute the first and second derivatives) or a simulation-based moment matching technique (\Correct{e.g.,} using random draws from a simple importance sampler).

\section{Scheme for Recycling all past simulated particles}\label{Sec_Recycling}		
After having proposed, in the previous section, a strategy in order to automatically specify the sequence of target distributions that will reduce the asymptotic variance of the estimator of the normalizing constant, we now focus on some strategies to improve the estimator of an expectation with respect to $\Target (\cdot)$.  

Using SMC samplers, this quantity is typically approximated with Eq. (\ref{FinalExpectationApproxSMC}), as:
\begin{equation}
J=\Expec_\Target \left[\TestFunction(\state) \right] = \int \Target (\State) \TestFunction (\State) d\State \approx \Expec_{\Target^N} \left[\TestFunction(\state) \right] = \sum_{i=1}^{\NbParticles } \NormISWeight_{\NbIterSMC }^{(i)} \TestFunction (\State_{\NbIterSMC }^{(i)}), 
\label{ClassicalExpecSMC}
\end{equation}
since $\Target _{\NbIterSMC }(\cdot)=\Target (\cdot)$. Only the samples from the iterations targeting the true posterior (generally only the last one) are taking into account for the approximation of the expectation. In this paper, in order to reduce the variance associated with this estimator in Eq. (\ref{ClassicalExpecSMC}), we propose two different strategies that will use particles drawn at the previous iterations of the sampler modified under a ``recycling principle''. We remark that these two recycling schemes are performed once the $T$ iterations of the SMC sampler are finished.

\subsection{Recycling based on Effective Sample Size}

\Correct{As discussed above, the  SMC approximation of the posterior expectation is typically only based on the samples from the last SMC iteration i.e. from $\pi_T$.} In order to have a more efficient estimation in the sense of minimizing the variance of the estimator in Eq. \eqref{ClassicalExpecSMC}, the idea we propose to explore in this section is to recycle all the particles that have been generated through the $\NbIterSMC$ iterations of the SMC sampler. This is challenging as intermediate samples from the sequence of distributions $\left\{\pi_t\right\}_{t = 1}^{T-1}$ do not target directly the posterior of interest $\pi_T$.

In \cite{Gramacy:2010gf}, a strategy has been proposed in order to recycle all the elements of the Markov chain obtained from a simulated tempering based MCMC algorithm. Here, we propose to adapt this approach to the $T$ collections of weighted samples given at each iteration of the SMC sampler. The idea is to correct each of these $T$ set of weighted random samples by using an importance sampling identity  to adjust these samples which are not drawn from the distribution of interest $\pi(\cdot)$.

More specifically, at the end of the $t$-th iteration of the SMC sampler, the weighted particle system approximates the target distribution $\pi_t(\cdot)$ as follows:
\begin{equation}
\pi_t^{N}(d\State) \approx \sum_{i=1}^{N} \NormISWeight_{t}^{(i)} \delta_{\State_t^{(i)}}(d\State) .
\end{equation}
However, in order to be able to use importance sampling identity, we need to have a set of unweighted samples from $\pi_t(\State)$. For this purpose, an unbiased resampling step that consists in selecting particles according to their importance weights can be used \cite{Kuensch2005}.  With a multinomial resampling scheme, we obtain a new collection 
\begin{equation}
\left\{ \StateT_t^{(i)} \right\}_{i=1}^{N}  \sim \pi_t(\State) ,
\end{equation}
where for $i=1,\ldots,N$
\begin{equation}
\StateT_t^{(i)}=\State_t^{(J_t^i)} \text{ ~~with~~ }  J_t^i \iid {\cal M}\left( \NormISWeight_{t}^{(1)}, \ldots,\NormISWeight_{t}^{(N)}\right) .
\label{ResamplingRecycling}
\end{equation}
Let us remark that if the resampling stage has been already performed at a specific iteration of the SMC sampler, the previous described steps are not necessary since the obtained samples are already \Correct{asymptotically} drawn from the target distribution $\pi_t(\cdot)$ (in this case, we set directly $\StateT_t^{(i)}=\State_t^{(i)}$ for $i=1,\ldots,N$). At the end of the SMC sampler, we have $\NbIterSMC$ collections of random samples drawn from each distribution of the targeted sequence. 

Since we  know the distribution from which these random samples $\left\{ \StateT_t^{(i)} \right\}_{i=1}^{N}$ are sampled,  an estimate of the expectation in (\ref{ClassicalExpecSMC}) can be obtained by using an importance sampling identity:
\begin{equation}
\widehat{h}_{\TimeIndex } = \sum_{j=1}^{\NbParticles  } \frac{w_{\WeightESS,\TimeIndex }(\StateT_{\TimeIndex }^{(j)})}{\sum_{i=1}^{\NbParticles } w_{\WeightESS,\TimeIndex }(\StateT_{\TimeIndex }^{(i)})} \TestFunction (\StateT_{\TimeIndex }^{(i)}) ,
\label{LocalESSestimator}
\end{equation}
with 
\begin{equation}
w_{\WeightESS,\TimeIndex }(\StateT_{\TimeIndex }^{(j)})=\frac{\UnNorTarget (\StateT_{\TimeIndex }^{(j)})}{\UnNorTarget _{\TimeIndex }(\StateT_{\TimeIndex }^{(j)})} ,
\label{CorrectedWeightsESS}
\end{equation}
where $\UnNorTarget(\cdot)$ and $\UnNorTarget_{\TimeIndex }(\cdot)$  are the unnormalized target distribution at the final iteration (\Correct{i.e.,} the posterior) and at the $t$-th iteration, respectively.

Finally, an overall estimator that will take into account all these estimators (or potentially a subset $\Omega$ among these $T$ estimates) can be obtained as follows:
\begin{equation}
\widehat{h} = \sum_{\TimeIndex \in \Omega} \lambda_{\TimeIndex }\widehat{h}_{\TimeIndex } ,
\label{CombiEstimator}
\end{equation}
where $0\leq \lambda_{\TimeIndex } \leq \sum_{ t \in \Omega} \lambda_{\TimeIndex }=1$.

As discussed in \cite{Gramacy:2010gf}, the combination coefficients $\lambda_{\TimeIndex }$ have to be chosen carefully if we do not want to have the variance of the estimator (\ref{CombiEstimator}) smaller than the one without recycling given in Eq. (\ref{ClassicalExpecSMC}). For example, a tempting solution is to take for $\TimeIndex \in \Omega$:
\begin{equation}
\lambda_{\TimeIndex }= \frac{W_{\WeightESS,\TimeIndex }}{W_{\WeightESS}}
\label{Naive}
\end{equation}
with $W_{\WeightESS,\TimeIndex }=\sum_{j=1}^{\NbParticles  } w_{\WeightESS,\TimeIndex }(\StateT_{\TimeIndex }^{(j)})$ and $W_{\WeightESS}=\sum_{\TimeIndex \in \Omega} W_{\WeightESS,\TimeIndex }$ but this can lead to very poor performance of the resulting estimator as illustrated empirically in the numerical simulation section in which we denote this choice by the ``naive'' recycling scheme. The solution proposed by \Correct{\cite{Gramacy:2010gf}} is thus to find all the $\lambda_{\TimeIndex }$ that maximizes the effective sample size of the weights of the entire population of particles. By combining Eqs. (\ref{CombiEstimator}) and (\ref{LocalESSestimator}), we can write:
\begin{equation}
\widehat{h} = \sum_{\TimeIndex \in \Omega} \sum_{j=1}^{\NbParticles  } \lambda_{\TimeIndex } \frac{w_{\WeightESS,\TimeIndex }(\StateT_{\TimeIndex }^{(j)})}{W_{\WeightESS,\TimeIndex }} \TestFunction (\StateT_{\TimeIndex }^{(i)}).
\label{CombiEstimatorEntirePopulation}
\end{equation}
The effective sample size of the entire population can then be defined as follows:
\begin{equation}
\ESS({\bm \lambda}_\Omega) = \left[ \sum_{\TimeIndex \in \Omega} \sum_{j=1}^{\NbParticles  } \left( \lambda_{\TimeIndex } \frac{w_{\WeightESS,\TimeIndex }(\StateT_{\TimeIndex }^{(j)})}{W_{\WeightESS,\TimeIndex }}\right)^2 \right]^{-1}.
\end{equation}
As a consequence, the set of coefficient ${\bm \lambda}_\Omega^*$ that maximize this effective sample size is the same as 
\begin{align}
		\begin{split}
		{\bm \lambda}_\Omega^* = \underset{{\bm \lambda}_\Omega  }{\argmin}  \quad \sum_{\TimeIndex \in \Omega} & \sum_{j=1}^{\NbParticles  } \left( \lambda_{\TimeIndex } \frac{w_{\WeightESS,\TimeIndex }(\StateT_{\TimeIndex }^{(j)})}{W_{\WeightESS,\TimeIndex }}\right)^2\\
		&\text{subject to } \sum_{\TimeIndex \in \Omega}  \lambda_{\TimeIndex }  =1.
		\end{split}
	\end{align}
By using Lagrangian multipliers, the optimal $\lambda_{\TimeIndex }^*$  within the SMC sampler framework  are defined, for $t\in \Omega$, by:
\begin{equation}
\lambda_{\TimeIndex }^*=\frac{l_{\TimeIndex }}{\sum_{n \in \Omega} l_n} \hspace{0.5cm}\text{ ~~with~~ } l_{\TimeIndex }=\frac{W_{\WeightESS,\TimeIndex }^2}{\sum_{j=1}^N w_{\WeightESS,\TimeIndex }(\StateT_{\TimeIndex }^{(j)})^2 }.
\label{ESSbased}
\end{equation}
Let us remark that the value $l_{\TimeIndex }$ involved in this optimal coefficients  $\lambda_{\TimeIndex }^*$  corresponds to the effective sample size of the $t$-th collection of importance weights given in (\ref{CorrectedWeightsESS}) and as a consequence $1\leq l_{\TimeIndex } \leq \NbParticles_{\TimeIndex }$.

\subsection{Recycling based on Deterministic Mixture Weights}
\label{DeMixRecyclingSection}

The previous solution is based on the combination of local estimators obtained by the collection of weighted particles from every iterations of the algorithm. In this section, we propose a new strategy that combines individual weighted particles by correcting their importance weights. This second scheme we propose is based on the principle, called the \textit{deterministic mixture} weight estimator proposed as in \cite{VeachGuibas95} and discussed by Owen and Zhou in \cite{Owen2000}. 

 This approach has been derived in order to combine weighted samples obtained from different proposal distributions in the importance sampler framework.  More recently, this technique has also been used in the Adaptive Multiple Importance Sampling (AMIS) of \cite{Cornuet2012} in order to recycle all past simulated particles in order to improve the adaptivity and variance of the Population Monte Carlo algorithm \cite{Cappe2004}. We propose to adapt this technique to the framework of the SMC sampler.

 As discussed in \cite{Owen2000}, using a deterministic mixture as a representation of the production of the simulated samples has the potential to exploit the most efficient proposals in the sequence $\PropDist_1 (\state ),\ldots,\PropDist_\NbIterSMC(\state )$ without rejecting any simulated value nor sample, while reducing the variance of the corresponding estimators. The poorly performing proposal functions are simply eliminated through the reduction of their weights and therefore their influence in:
 \begin{equation}
 \frac{\Target (\State_{\TimeIndex }^{(i)})}{\sum_{n=1}^{\NbIterSMC} c_n \PropDist_n(\State_{\TimeIndex }^{(i)})} ,
 \end{equation}
as $T$ increases (with $c_n=\NbParticles_n/\sum_{\TimeIndex =1}^{\NbIterSMC} \NbParticles_{\TimeIndex }$ is the proportion of particles drawn from the proposal $\PropDist _n$)\footnote{Here we assume the general case in which a different number of particles could be drawn at each iteration of the SMC sampler.}. Indeed, if $\PropDist _1$ is the poorly performing proposal, while the $\PropDist _n$'s ($n>1$) are good approximations of the target $\Target$, for a value  $\State_1^{(i)}$ such that $\Target(\State_1^{(i)})/\PropDist_1(\State_1^{(i)})$ is large, because  $\PropDist_1(\State_1^{(i)})$ is small (and not because it is a sample with high posterior value), $\Target (\State_{\TimeIndex }^{(i)})\diagup \{ c_1 \PropDist_1(\State_1^{(i)}) + \ldots + c_\NbIterSMC \PropDist_\NbIterSMC(\State_1^{(i)})  \}$ will behave like $\Target (\State_{\TimeIndex }^{(i)})\diagup \{ c_2 \PropDist_2(\State_1^{(i)}) + \ldots + c_\NbIterSMC \PropDist_\NbIterSMC(\State_1^{(i)})  \}$ and will decrease to zero as T increases.

 In our case, since we are not in the importance sampling framework with well defined proposal distribution but instead with  $\NbIterSMC$ collections of samples from the intermediate target distributions $\left( \left\{ \StateT_1^{(i)} \right\}_{i=1}^{\NbParticles_1},\ldots, \left\{\StateT_\NbIterSMC^{(i)}\right\}_{i=1}^{\NbParticles_\NbIterSMC}\right)$ by following the same resampling step as described in the previous section in Eq. (\ref{ResamplingRecycling}), the estimator of an expectation using this proposed deterministic mixture will be given by:
\begin{equation}
\Expec_\Target \left[\TestFunction \right] \approx \sum_{\TimeIndex =1}^{\NbIterSMC}\sum_{i=1}^{\NbParticles  _{\TimeIndex }} \frac{{w}_{\WeightDemix,\TimeIndex }^{(i)}}{\sum_{k =1}^{\NbIterSMC}\sum_{j=1}^{\NbParticles  _{k }} \hat{w}_{\WeightDemix,k }^{(j)} } \TestFunction (\StateT_{\TimeIndex }^{(i)}) ,
\end{equation}
with 
\begin{equation}
 {w}_{\WeightDemix,\TimeIndex }^{(i)} = \frac{\Target (\StateT_{\TimeIndex }^{(i)})}{\sum_{n=1}^{\NbIterSMC} c_n \Target _n(\StateT_{\TimeIndex }^{(i)})} ,
 \label{SecondSolution}
\end{equation}
where $c_n=\NbParticles_n/\sum_{\TimeIndex =1}^{\NbIterSMC} \NbParticles_{\TimeIndex }$ is the proportion of particles drawn from $\Target _n$ amongst all the simulated particles. The problem with this strategy is we need to evaluate the target $\Target _{\TimeIndex }(\cdot)$ exactly (not up to a constant) and thus we need to know the normalizing constant $\NormConst_{\TimeIndex }$ involved in all the intermediate target distributions $\Target _{\TimeIndex }(\cdot)=\gamma_{\TimeIndex }(\cdot)/\NormConst_{\TimeIndex }$. Hence, at his point the idea we propose is to use the (unbiased) SMC approximation of each normalizing constant given by Eq. (\ref{Eq_Approx_Ratio_NormalizingConstantZ0Zt}). As a consequence, the weights of this proposed recycling scheme, defined originally in  Eq (\ref{SecondSolution}), is thus approximated by:
\begin{equation}
 {w}_{\WeightDemix,\TimeIndex }^{(i)} \approx \frac{\gamma(\StateT_{\TimeIndex }^{(i)})}{\sum_{n=1}^{\NbIterSMC} c_n \gamma_n(\StateT_{\TimeIndex }^{(i)})\hat{\NormConst }_n^{-1}} .
\end{equation}

\section{Numerical Simulations}\label{Sec_Simus}	

In this section, the performances of the proposed strategies used to improve the SMC sampler-based estimators are assessed through two different models and also for a class of models used widely in signal processing, namely penalized regression. In the remainder of this paper, we adopt a parametric form for the temperature schedule: $\CoolSchedule  _{\TimeIndex }=h( \TimeIndex;\ParametricParameter ,\NbIterSMC )$, which satisfies the following conditions: $\{\CoolSchedule  _{\TimeIndex }\}$ is non-decreasing function, $\CoolSchedule  _{0}=0$ and $\CoolSchedule  _{\NbIterSMC  }=1$.  This efficiently reduces the optimization problem to a univariate problem of finding the optimal value for an unique parameter $\ParametricParameter$ instead of $\NbIterSMC-1 $ parameters $\{\IncreStep_{\TimeIndex } \}_{\TimeIndex =2}^{\NbIterSMC }$.	The parametric function used for the proposed adaptive cooling schedule strategy  is defined as:
		\begin{equation}
	\phi_t=\frac{\exp(\gamma t/T)-1}{\exp(\gamma)-1}.
	\label{Chap2ParametricCooling}
	\end{equation}

\subsection{Model 1: Linear and Gaussian Model}
	
Let us firstly consider a linear and Gaussian model for which the a posteriori distribution as well as the marginal likelihood can be derived analytically. \Correct{The} proposed strategies \Correct{can} thus be compared with the optimal Bayesian inference method. \Correct{More precisely, we assume}

\begin{align}
\begin{split}
p(\State) & =  \Normal (\State|{\bm \mu}, {\bm \Sigma}),  \\
p(\Data|\State) & =  \Normal (\Data|{\bm H}\State, {\bm \Sigma}_y) .
\end{split}
\end{align}

For this model, the posterior distribution is given by $p(\State|\Data)=\Normal (\State|{\bm \mu}_p, {\bm \Sigma}_p)$ 
with 
\begin{align}
\label{TruePosteriorMeanNormal}
\begin{split}
{\bm \mu}_p &= {\bm \mu}+{\bm \Sigma} {\bm H}^T \left( {\bm H} {\bm \Sigma} {\bm H}^T + {\bm \Sigma}_y\right)^ {-1} \left[\Data- {\bm H}{\bm \mu} \right] ,\\
{\bm \Sigma}_p &=  \left( {\bm I}_\DimState -  {\bm \Sigma} {\bm H}^T \left( {\bm H} {\bm \Sigma} {\bm H}^T + {\bm \Sigma}_y\right)^ {-1}  {\bm H} \right)  {\bm \Sigma}.
\end{split}
\end{align}
In addition, the marginal likelihood (i.e. the normalizing constant) is also know in closed form:
\begin{equation}
p(\Data) = \Normal ( \Data|{\bm H}{\bm \mu} , {\bm H} {\bm \Sigma} {\bm H}^T + {\bm \Sigma}_y).
\end{equation}

\vspace*{0.3cm}

For illustration, we select ${\bm \Sigma}=10{\bm I}_{10}$ and ${\bm \mu}={\bm 0}_{10\times 1}$ for the prior distribution. Concerning the likelihood parameters, all the elements of the transition matrix have been randomly generated using a standard normal distribution and $ {\bm \Sigma}_y={\bm I}_{\DimObs}$ with a varying number of observations $\DimObs$. \Correct{Regarding to} the SMC sampler, and in particular for the adaptive MWG (summarized in Algo. \ref{Algo_AMWG}), we use for the forward kernel: $\NbMCMCmove=5$ and $B=5$.

\subsubsection{Analysis of the proposed \Correct{adaptive} cooling schedule}

Under this model, the proposed approach we develop is optimal (in the sense of minimizing the asymptotic variance of the normalizing constant), since each intermediate target distribution is a multivariate normal distribution. In Fig. \ref{fig:Chap2:EvolAsympVariance}, the evolution of the theoretical asymptotic variance of the normalizing constant \Correct{estimator} defined in Eq. (\ref{OptimalVar_ResamplingBefore}) with the parameter value $\ParametricParameter$ is depicted for the SMC sampler as a function of the number of iterations $T$ and for different number of observations. We can clearly see that an optimal value exists for the parametric function of the cooling schedule, that will minimize the asymptotic variance. 

In Fig. \ref{fig:Chap2:CheckAsympVariance}, we compare the theoretical asymptotic variances of the normalizing constant with the ones obtained by simulation. In order to obtain these results, we have run 500 times an SMC sampler that utilizes a perfect mixing forward kernel which can  be straightforwardly obtained  analytically for this specific model, i.e.:
\begin{eqnarray}
K_{\TimeIndex}(\state_{\TimeIndex-1},\state_{\TimeIndex})& =&\pi_{\TimeIndex}(\state_{\TimeIndex})\propto p(\Data|\State)^{\CoolSchedule_{\TimeIndex}} p(\State) \nonumber \\
&=& \Normal ( \State|{\bm \mu}_{\TimeIndex} ,{\bm \Sigma}_{\TimeIndex} ) ,
\label{PerfectMixingNormalCase}
\end{eqnarray}
with 
\begin{equation}
{\bm \mu}_{\TimeIndex} = {\bm \mu}+{\bm \Sigma} {\bm H}^T \left( {\bm H} {\bm \Sigma} {\bm H}^T + \frac{1}{\CoolSchedule_{\TimeIndex}} {\bm \Sigma}_y\right)^ {-1} \left[\Data- {\bm H}{\bm \mu} \right] , 
\end{equation}
\begin{equation}
{\bm \Sigma}_{\TimeIndex}  =  \left( {\bm I}_\DimState -  {\bm \Sigma} {\bm H}^T \left( {\bm H} {\bm \Sigma} {\bm H}^T +\frac{1}{\CoolSchedule_{\TimeIndex}} {\bm \Sigma}_y\right)^ {-1}  {\bm H} \right)  {\bm \Sigma} .
\end{equation}
From \Correct{Fig. \ref{fig:Chap2:CheckAsympVariance}}, we can see that the variance of the normalizing constant \Correct{estimator for the} SMC sampler with \Correct{a} finite number of particles is very close to the theoretical asymptotic value. Indeed, only few  particles are required to reach these asymptotic variances under this model.

\begin{figure}[!htb]
\centering
\subfloat[20 Observations]{
\includegraphics[width=0.5\textwidth]{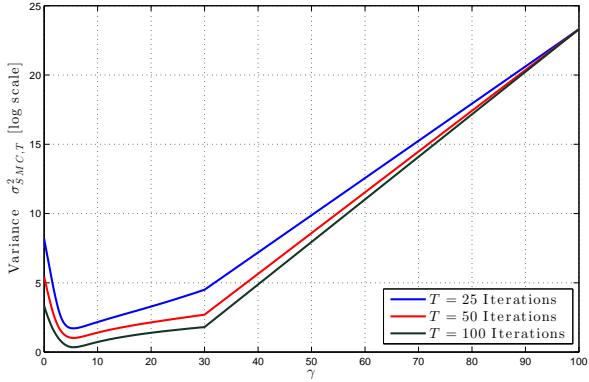} \label{fig:Chap2:EvolAsympVariance20Obs}}
\subfloat[30 Observations]{
\includegraphics[width=0.5\textwidth]{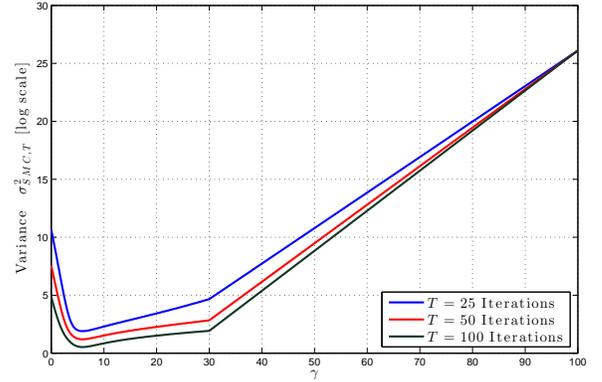} \label{fig:Chap2:EvolAsympVariance30Obs}}
\caption{Evolution of the theoretical asymptotic variance of the SMC sampler estimate of the normalizing constant versus the value of $\gamma$ in the cooling schedule as the function of the numbers of iterations for different number of observations}
\label{fig:Chap2:EvolAsympVariance}
\end{figure}

\begin{figure}[!htb]
\centering
\subfloat[20 Observations]{
\includegraphics[width=0.5\textwidth]{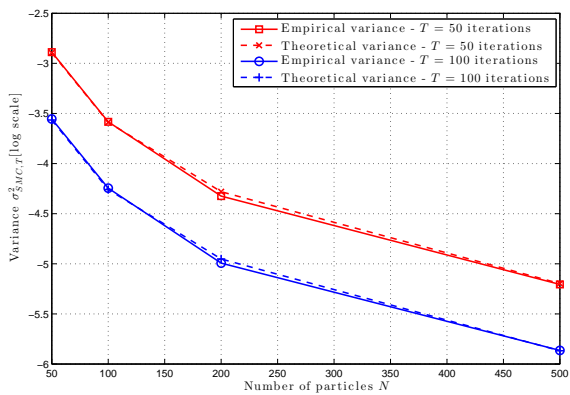} }
\subfloat[30 Observations]{
\includegraphics[width=0.5\textwidth]{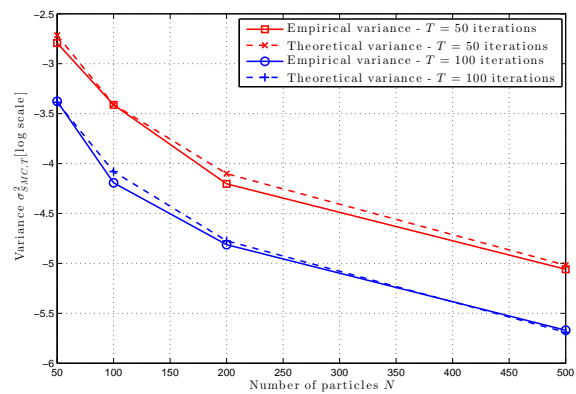} }
\caption{Comparison of the theoretical asymptotic variances (dashed lines) and the empirical ones from SMC sampler using perfect mixing Markov Kernel (solid lines) by using the optimal value of the parameter $\hat{\ParametricParameter}$ with 20 Observations.}
\label{fig:Chap2:CheckAsympVariance}
\end{figure}

Fig. \ref{fig:Chap2:ComparisonVariance} compares the proposed methodology for the optimal adaptive cooling schedule versus two others competitors, the one based on the $\CESS$ of  \cite{Yan2013} and the linear cooling schedule often used in practice. Performances of the SMC samplers are illustrated for two choices of mutation kernels: the perfect mixing kernel (Eq. \ref{PerfectMixingNormalCase}) or the adaptive random walk Metropolis Hastings kernel. These results clearly show the benefit of using such adaptive cooling schedule  - a bad choice can lead to a very poor estimate which has a large variance. Variance results obtained from the proposed approach and the $\CESS$-based approach are comparable. The main advantage of our proposed approach is that we control the global complexity of the SMC sampler since we set the number of iterations of the SMC sampler whereas in the $\CESS$-based strategy, the number of iterations of the SMC samplers will depend on the problem under consideration as well as the predefined value of $\CESS$. In order to be able to compare  both approaches with the same complexity, several runs of the SMC sampler with different values of the $\CESS$ have been performed to obtain the  $\CESS$ value that roughly leads to a specific number of iterations $\NbIterSMC$ (25, 50 and 100). 

\begin{figure}[!htb]
\centering
\subfloat[\footnotesize Perfect Mixing Kernel]{
\includegraphics[width=0.5\textwidth]{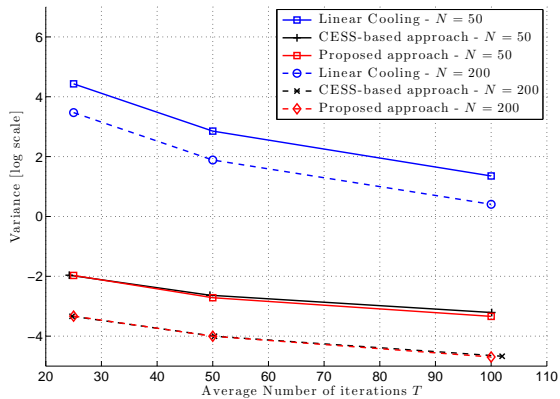}}
\subfloat[\footnotesize Adaptive MWG kernel]{
\includegraphics[width=0.5\textwidth]{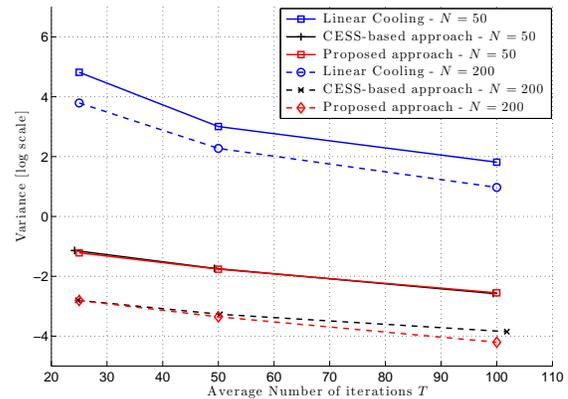}}
\caption{Comparison of the different cooling schedule strategies in terms of the variance of the normalizing constant estimate for different number of particles. Results are obtained with the use of either the perfect mixing Kernel (left) or the adaptive MWG kernel (right).}
\label{fig:Chap2:ComparisonVariance}
\end{figure}

\subsubsection{Analysis of the proposed recycling schemes}

We finally assess the performance of the two proposed recycling schemes. In order to analyze the potential gain of recycling past simulated particles, four different estimators based on the output of the SMC sampler are compared: \textit{``no recycling''} given in Eq. (\ref{ClassicalExpecSMC}),  \textit{``Na\"ive recycling''} and  \textit{``ESS-based recycling''} given in Eq. (\ref{CombiEstimator}) with $\lambda_{t}$ defined respectively in Eq (\ref{Naive}) and(\ref{ESSbased}) and the \textit{``DeMix-based recycling''} described in Section \ref{DeMixRecyclingSection}.

Fig.  \ref{fig:Chap2:MSE_PosteriorMean_Normal} shows the mean squared error (MSE) between the estimated posterior mean and the true posterior mean given by  ${\bm \mu}_p$ in Eq. (\ref{TruePosteriorMeanNormal}). We can firstly remark from these results that the na\"ive recycling scheme does not really improve the performance of the estimator of the posterior mean. On the contrary, both our proposed \Correct{schemes outperform} significantly this na\"ive recycling and the classical estimator that uses only the final population of particles (No recycling scheme). The improvement increases with the number of iterations used in the SMC sampler, as expected since more collection of particles can be recycled in the estimator. These results demonstrate also empirically for this model the superiority of the DeMix recycling approach.

\begin{figure}[!htb]
\centering
\subfloat[20 Obs.]{
\includegraphics[width=0.5\textwidth]{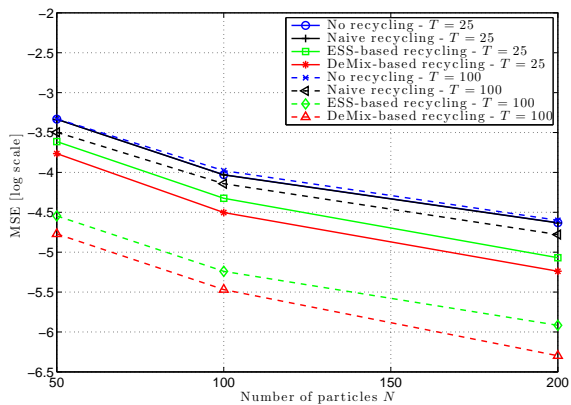}}
\subfloat[40 Obs.]{
\includegraphics[width=0.5\textwidth]{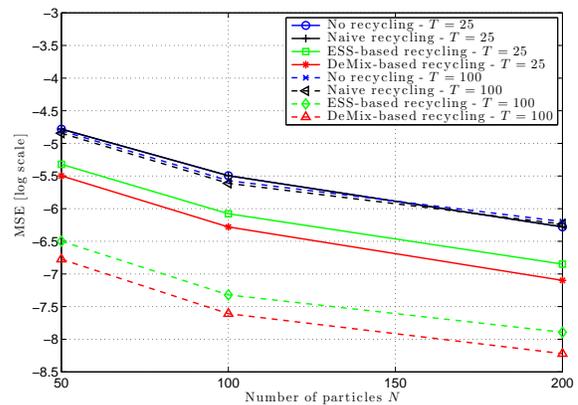}}
\caption{Mean square error between the estimated and the true posterior mean for Model 1 using the different recycling schemes}
\label{fig:Chap2:MSE_PosteriorMean_Normal}
\end{figure}

\subsection{Model 2: Multivariate Student's t Likelihood}
	
	In this second example, we illustrate the results of our optimal schedule and recycling methods with a multivariate Student's t distribution as likelihood:
\begin{eqnarray}
p(\State) & = & \Normal (\State|{\bm \mu}, {\bm \Sigma}) \nonumber \\
p(\Data|\State) & = &\frac{\Gamma\left(\frac{\nu + \DimObs}{2}\right)}{\Gamma\left(\frac{\nu}{2}\right)(\nu \pi)^{\DimObs/2}}|{\bm \Sigma}_l|^{-\frac{1}{2}}\left[1 + \frac{[\bm{y} - {\bm H}\State]^T{\bm \Sigma}_l^{-1}[\bm{y} - {\bm H}\State]}{\nu}\right]^{-\frac{(\nu + \DimObs)}{2}} \nonumber
\label{ModelStudentT}
\end{eqnarray}

This model can be particularly challenging due to the possibility of having a multimodal target posterior when contradictory observations are used. To analyze the performance of the proposed scheme in complex situation, we use 
$${\bm H}=\begin{bmatrix}
1 & 1 & 0 & 0 \\
0 & 0 & 1 & 1 
\end{bmatrix}^T$$
 ${\bm \Sigma}=20 {\bm I}_2$, ${\bm \mu}={\bm 0}_{2\times 1}$, ${\bm \Sigma}_l=0.1 {\bm I}_4$ and we observe a vector $\Data=\begin{bmatrix}\OneData_{1} & \OneData_{2} & \OneData_{3}  & \OneData_{4}\end{bmatrix}^T=\begin{bmatrix}8 & -8 & 8 & -8\end{bmatrix}^T$. These particular choices lead to an highly  multimodal posterior distribution as illustrated in Fig. \ref{fig:Chap2:PosteriorStudent} for two different values of the degree of freedom of the multivariate Student's t likelihood. From \Correct{this model} and the  parameters used, we have for both values of the degree of freedom $\Expec_\Target \left[ \State \right]=[0 ~ 0]^T$,
which is confirmed by the numerical evaluation of the posterior shown in Fig. \ref{fig:Chap2:PosteriorStudent}. For this model, we will follow the same procedure as in the previous Model: analysis of proposed adaptive cooling schedule and then of the proposed recycling schemes. In all the numerical simulations presented in this section, we have chosen  $\NbMCMCmove=10$ and $B=2$ as parameters of the adaptive MWG kernel within the SMC sampler.

	\begin{figure}[!htb]
\centering
\subfloat[$\nu=0.2$]{
\includegraphics[width=0.5\textwidth]{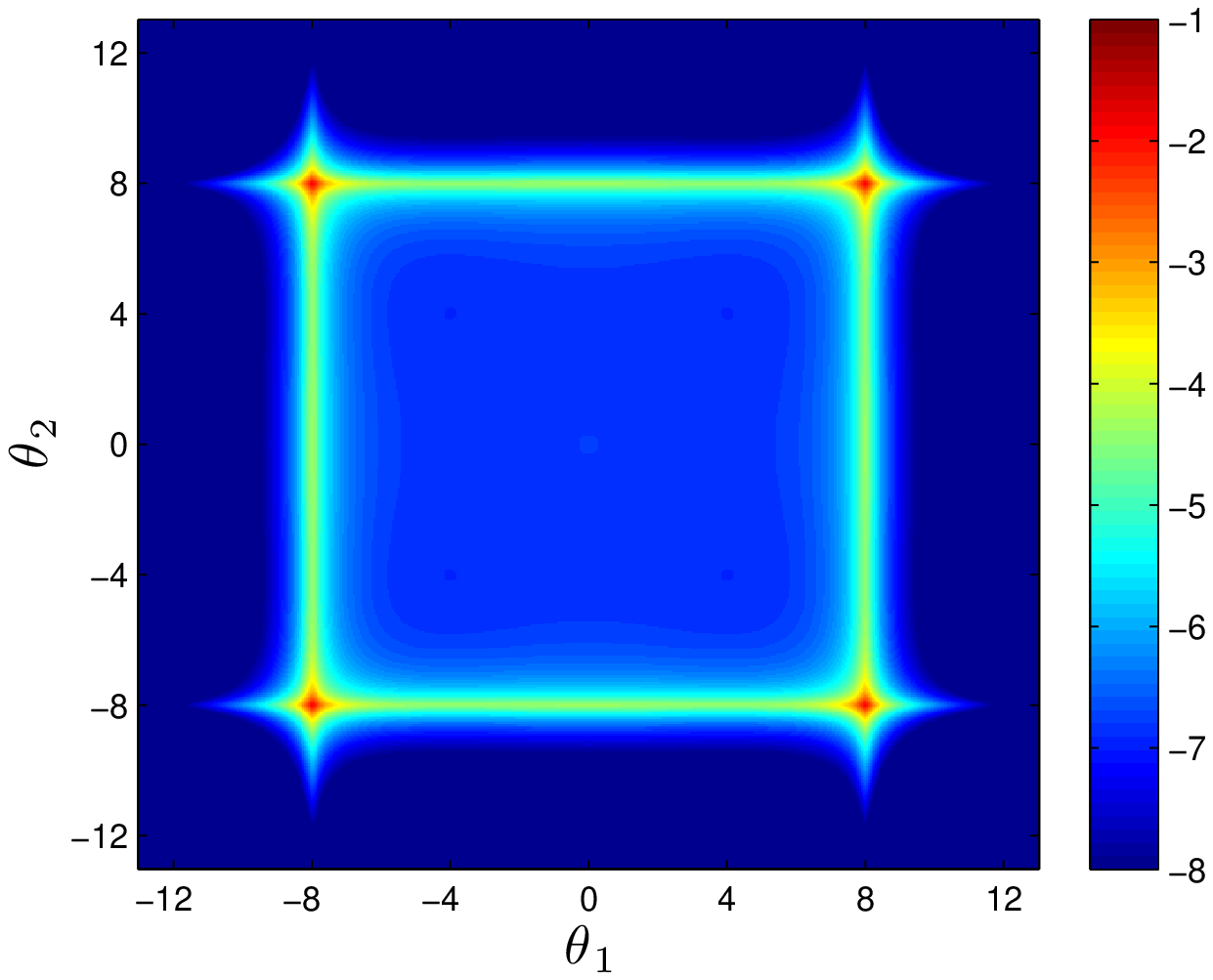} \label{fig:Chap2:PosteriorStudentSmall}}
\subfloat[$\nu=7$]{
\includegraphics[width=0.5\textwidth]{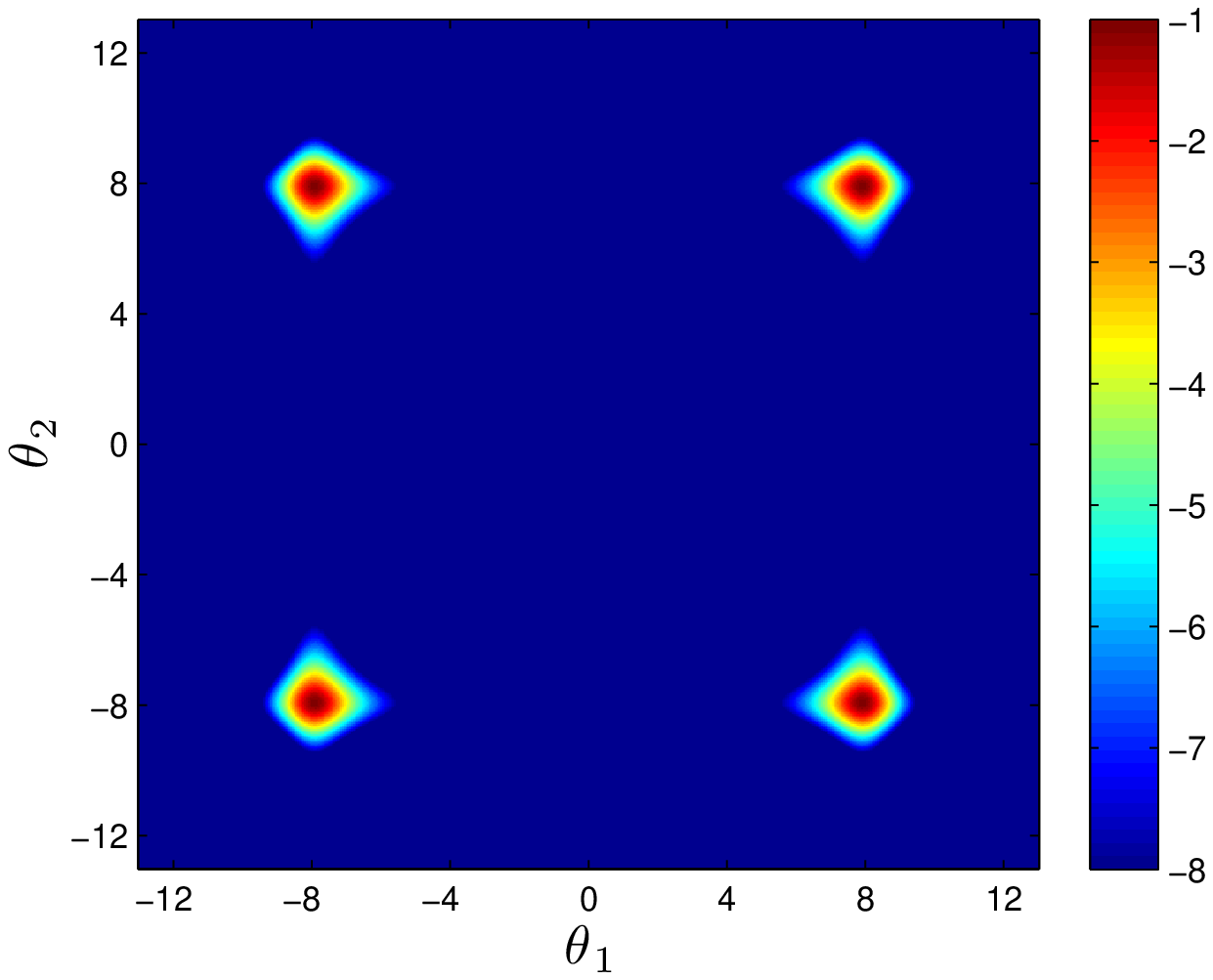} \label{fig:Chap2:PosteriorStudentLarge}}
\caption{Target posterior distribution $p(\State|\Data)$ in log scale evaluated on a grid with 2 different values for the degree of freedom of the Student's t likelihood}
\label{fig:Chap2:PosteriorStudent}
\end{figure}

Table \ref{table:chap2:VarNormConstSmallDegree} shows the variance of the estimated normalizing constant (\Correct{i.e.,} $p(\Data)$)  when the degree of freedom of the multivariate Student's t distribution is $\nu=0.2$ and $\nu=7$, respectively. We compare the results obtained using different cooling schedules. The proposed adaptive approach, the  $\CESS$-based one as well as the linear cooling schedule \Correct{yield} similar results. From our \Correct{simulation} results, we can see that the proposed adaptive procedure takes a very small value (close to 0) as optimal value for $\gamma$ which thus \Correct{leads to}  the linear cooling schedule. Same remark when we analyze the evolution of the temperature given by the ``on-line'' $\CESS$-based strategy. Nevertheless, we can see that these variances can degrade very significantly if another value of $\gamma$ is chosen (here we take $\gamma=6$). This clearly demonstrates the impact of this temperature schedule in terms of the variance of the normalizing constant. The proposed procedure is thus of great interest in order to automatically decide what should be the evolution of this cooling schedule for a given number of SMC iterations.

\begin{table}[!htb]
\begin{footnotesize}
\begin{center}
\begin{tabular}{cc|c|c|c|c|c|c|c|c|c}
& & \multicolumn{2}{c|}{Linear}  & \multicolumn{2}{c|}{Cooling}   &  \multicolumn{2}{c|}{CESS}  & \multicolumn{2}{c}{Proposed}\\
&& \multicolumn{2}{c|}{Cooling} &   \multicolumn{2}{c|}{$\gamma=6$} & \multicolumn{2}{c|}{Approach} & \multicolumn{2}{c}{Approach} \\
&& \multicolumn{2}{c|}{$\gamma\rightarrow 0$} &  \multicolumn{2}{c|}{} &  \multicolumn{2}{c|}{} & \multicolumn{2}{c}{} \\
&& $\nu=0.2$ & $\nu=7$ &$\nu=0.2$ & $\nu=7$ & $\nu=0.2$ & $\nu=7$ & $\nu=0.2$ & $\nu=7$ \\
\hline
 & $N=50$ & 0.0026 & 0.0146 & 0.0110 & 0.0375 & 0.0028 & 0.0177 & 0.0030 & 0.0209 \\
 \multirow{1}{*}{$T=$25 Iter.}   & $N=100$ & 0.0013& 0.0086 & 0.0055&0.0152 & 0.0012 & 0.0079 & 0.0015  & 0.0088 \\   
    & $N=200$ &  0.0006 &  0.0050 &  0.0024  &  0.0090 & 0.0007 & 0.0041 &0.0008  & 0.0042\\ 
    \hline
 & $N=50$ & 0.0011 & 0.0105 &   0.0046  &  0.0160  & 0.0016  & 0.0078 & 0.0013 & 0.0072 \\   
 \multirow{1}{*}{$T=$50 Iter.}   & $N=100$ & 0.0007 & 0.0050  &  0.0029  &  0.0102 & 0.0006 & 0.0037   & 0.0006 & 0.0039 \\   
    & $N=200$ &  0.0003   &  0.0028 &  0.0012  &  0.0043    &  0.0004  &  0.0025 & 0.0004 & 0.0017\\  
        \hline
 & $N=50$ &  0.0006 &  0.0047  &   0.0026 &   0.0078 & 0.0006& 0.0051 & 0.0009& 0.0037\\  
 \multirow{1}{*}{$T=$100 Iter.}  & $N=100$ & 0.0003 &  0.0022 &   0.0013 &   0.0044  & 0.0003  & 0.0022  & 0.0004 & 0.0023  \\  
    & $N=200$ &   0.0002 &   0.0016 &   0.0005 &  0.0026  & 0.0002 & 0.0010 &  0.0002  &  0.0013\\ 
\end{tabular}
\end{center}
\end{footnotesize}
\caption{Comparison of the variance of the normalizing constant \Correct{estimator} obtained by using different cooling schedules for Model 2 with $\nu=0.2$ and $\nu=7$.}
\label{table:chap2:VarNormConstSmallDegree}
\end{table}

\Correct{Fig. \ref{fig:Chap2:MSE_PosteriorMean_Student}} shows the mean squared error between the estimated posterior mean from the proposed recycling scheme and the true one. Unlike the previous model (linear and Gaussian one) \Correct{for} $\nu=0.2$,  the na\"ive recycling outperforms the classical estimator of the SMC sampler when only the last collection of particles is used. This could be explained by the shape of the target posterior (Fig. \ref{fig:Chap2:PosteriorStudent}). Indeed, in such a case, the posterior has a large region with \Correct{a} non-zero probability  in the middle of the ``square''. As a consequence, the particles of the first iteration of the SMC sampler that target the prior \Correct{can} be very useful. However, when the degree of freedom of the likelihood is high ($\nu=7$), this remark does not hold since the posterior is really concentrated on 4 modes. From this case, it is also interesting to see that \Correct{the} MSE increases with the number of iterations used in the SMC sampler when either no recycling or na\"ive recycling is performed. Indeed, by increasing the number of iterations, we increase also the number of potential resampling steps and we know that during the resampling procedure, some particles which are currently \Correct{located} in one of the 4 modes can be discarded. Therefore it becomes very difficult for the SMC sampler to jump between two well separated modes, thus leading to an unexplored mode by the SMC sampler for the next iteration. This effect does not appear with the proposed recycling scheme since we recycle all the past simulated particles.

	\begin{figure}[!htb]
\centering
\subfloat[$\nu=0.2$]{
\includegraphics[width=0.5\textwidth]{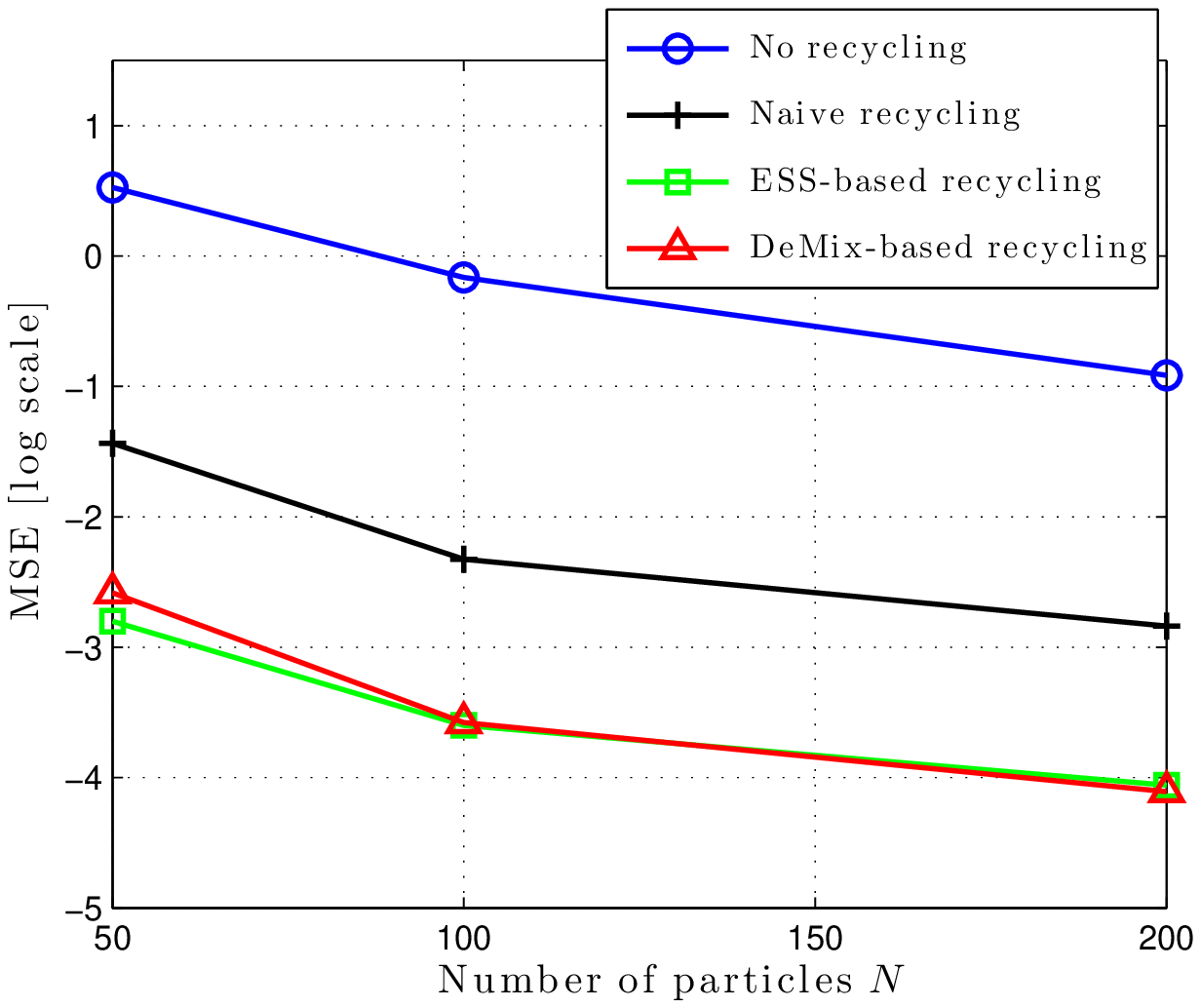}} 
\subfloat[$\nu=7$] {
\includegraphics[width=0.5\textwidth]{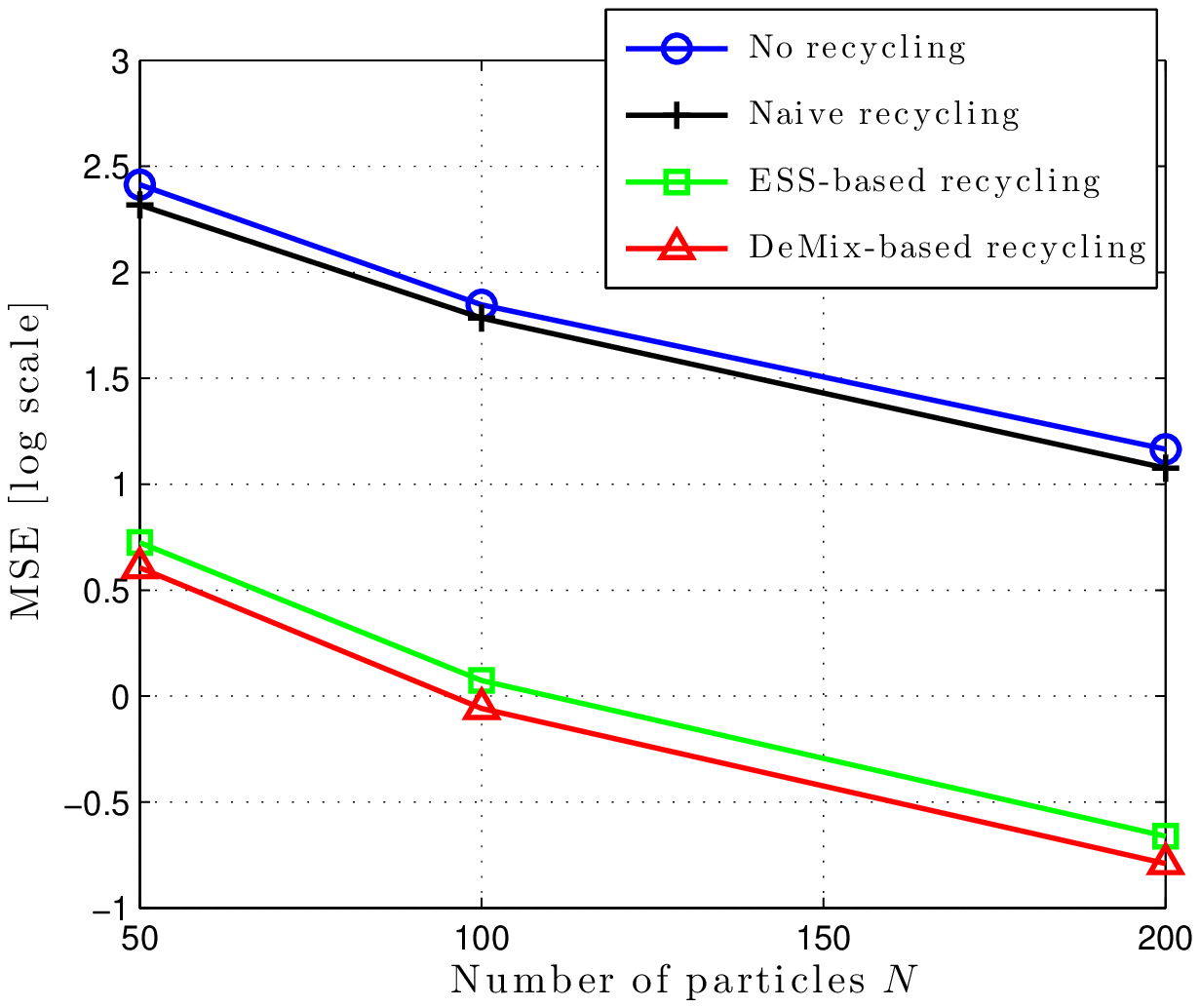}}
\caption{Mean squared error between the estimated and the true posterior mean for Model 2 using the different recycling schemes with $T=$100 Iterations.}
\label{fig:Chap2:MSE_PosteriorMean_Student}
\end{figure}

	Finally, in order to emphasize the significant gain that could be obtained using our proposed recycling schemes, Table \ref{tablechap2KSdistanceSmallDegree} shows the mean and standard deviation of the Kolmogorov-Smirnov distance defined as 
	$D=\sup_{\theta_1} \left| F^N(\theta_1) -  F(\theta_1) \right|$ 
	where $F^N$ and $F$ are the empirical cumulative distribution obtained from the SMC sampler and the true posterior cumulative distribution, respectively. This distance $D$ is obtained through 100 runs of the SMC samplers. Compared to the previous \Correct{comparisons related to}  the MSE of the posterior mean, this measure give us some information about the quality of the approximation of the whole target distribution. In order to obtain \Correct{these} results, the true target cumulative distribution $F(\theta_1)$ has been obtained numerically by using a very fine grid. In both cases ($\nu=0.2$ and $\nu=7$), these results empirically demonstrate the significant gain obtained by using the proposed recycling schemes with a slight advantage to the DeMix-based approach. The average and the standard deviation of this Kolmogorov-Smirnov distance are divided by a factor of 2-3 compared to the case in which we use only the collection of particles from the last iteration of the SMC sampler.
	
	\begin{table}[!htb]
\centering
\begin{footnotesize}
\begin{tabular}{cc|c|c|c|c}
& & No Recycling  & Naive    &  ESS-based  & DeMix\\
& &  						&				Recycling				& Recycling & Recycling\\
\hline
 & $N=50$ &    0.1276    (0.0460) & 0.0727    (0.0234) & 0.0458    (0.0121) & 0.0407    (0.0123) \\
$T=$25 Iter.  & $N=100$ &   0.0835    (0.0224) & 0.0488    (0.0164) & 0.0366    (0.0094) & 0.0315    (0.0089) \\ 
  & $N=200$ &   0.0615    (0.0182) & 0.0353    (0.0103) & 0.0254    (0.0055) & 0.0237    (0.0053) \\
  \hline
  & $N=50$ &   0.1274    (0.0424) & 0.0514    (0.0167) & 0.0357    (0.0096) & 0.0311    (0.0097) \\
$T=$50 Iter.  & $N=100$ &   0.0898    (0.0251) & 0.0379    (0.0117) & 0.0268    (0.0067) & 0.0230    (0.0055) \\
  & $N=200$ &   0.0627    (0.0188) & 0.0267    (0.0068) & 0.0201    (0.0045) & 0.0185    (0.0037) \\
  \hline
  & $N=50$ &   0.1186    (0.0352) & 0.0391    (0.0117) & 0.0315    (0.0066) & 0.0243    (0.0060) \\
$T=$100 Iter.  & $N=100$ &   0.0846    (0.0231) & 0.0288    (0.0079) & 0.0226    (0.0054) & 0.0187    (0.0038) \\
  & $N=200$ &   0.0599    (0.0188) & 0.0216    (0.0054) & 0.0177    (0.0033) & 0.0159    (0.0031)
	\end{tabular}
	\end{footnotesize}

\caption{Comparison of recycling schemes for the accuracy to approximate the posterior distribution $p(\theta_1|\Data)$ in terms of  the Kolmogorov-Smirnov distance (mean and standard deviation in parentheses) for Model 2 with $\nu=0.2$.}
\label{tablechap2KSdistanceSmallDegree}
	\end{table}

\subsection{Penalized regression model with count data}

In this section, we illustrate how the proposed SMC sampler can be efficiently used  in penalized regression models with particular focus on counting process observations. Sparse regression analysis initially studied in the context of penalized least squares or likelihood has gained increasing popularity since the seminal paper on the LASSO \cite{Tibshirani:1996wb}. Since this work, many approaches under both frequentist and Bayesian have been proposed to extend these sparsity inducing regression frameworks.

In a frequentist setting the most common choice sparsity inducing penalty is the $\lone$-regularization for the regression coefficients $\Coefficient \in \reel^p$ and it is known as LASSO, with penalty term $\gamma \sum_{i=1}^p |\beta_i|$. Under a Bayesian modelling paradigm, in which the regression coefficients are treated as a random vector, one may recover the LASSO estimates from the maximum a posteriori (MAP) point estimator of the coefficients via a choice of prior on the coefficients given by the multivariate Laplace distribution, $p(\Coefficient) \propto  exp(- \gamma \sum_{i=1}^p |\beta_i|)$.
 
A limitation in this approach is the use of identical penalization on each regression coefficient. This can lead to unacceptable bias in the resulting estimates \cite{Fan:2001ug}. Indeed, the classical $\lone$-regularization can lead to an over-shrinkage of large regression coefficients even in the presence of many zeros. This has resulted in sparsity-inducing non-convex penalties that use different penalty coefficients on each regression coefficient, i.e. $\sum_{i=1}^p \gamma_i |\beta_i|$  have been proposed, as have grouping regularization constraints, see adaptive and sequential estimation approaches in \cite{Zou:2006wm,Lee:2012vx,Candes:2008ua,Chartrand:2008ti}. Alternative non-convex approaches include the bridge regression framework, i.e. $\gamma \sum_{i=1}^p |\beta_i|^q$ with $q\in (0,2)$ which is obtained using the exponential power (EP) distribution:
\begin{equation}\label{DefinitionExponentialPower}
f(\Coefficient;\gamma,q)=\prod_{i=1}^{\NbCoefficient}\frac{q}{2\gamma\Gamma(1/q)} \exp \left( -\left| \frac{\OneCoefficient_i}{\gamma} \right|^q \right),
\end{equation}
which leads to the $\lqreg$-regularization problem \cite{Polson:2011wb,LeThuNguyen:2013vl}. Compared to previous non-convex prior, the latter possesses the advantage of not introducing additional tuning variables that need to be selected.

More specifically, in this section, we address the challenging regression problem for which the proposed strategies can be used as an efficient solution in finding the relationship between continuous input variables and count data as response. The likelihood is given by Poisson distribution defined as:

			\begin{equation}
				\OneData_{\RowMatrix } \sim {\cal P}o(\OneData_{\RowMatrix }|\OneMeanY_{\RowMatrix }) \qquad \text{with } \qquad \OneMeanY_{\RowMatrix }=\exp\left (\OneCoefficient_{0}+\sum_{\ColumnMatrix=1}^{\NbCoefficient}\OneCoefficient_{\ColumnMatrix}\DesignMatrix_{k}^j\left(\Covariates_{\RowMatrix ,\ColumnMatrix }\right)\right ) 
			\end{equation}	
			and $\DesignMatrix_{k}^j\left(\cdot\right)$ corresponds to the basis function used in the regression.

			For the experiments, the true $\DimObs=100$ observations have been generated with regression coefficients set to zeros except $\OneCoefficient_0=1$, $\OneCoefficient_{2}=1.5$, $\OneCoefficient_{4}=-2$, $\OneCoefficient_{6}=1$, $\OneCoefficient_{7}=-2$ and $\OneCoefficient_{9}=1.2$. For the basis function, a Gaussian kernel defined as $\designmatrix_{i,\ColumnMatrix}(\Covariates_{i})=\exp \left \{ \frac{\| (\Covariates_{i}  -\vect{c}_j\|_{2}^{2}}{r_j^{2}}\right \}$ 				
			with $11$ equally spaced centers $c_{\ColumnMatrix}$  with \Correct{the same scale parameters  $r_{\ColumnMatrix}=r=0.5$ have been used}. 
			
			The dimension of the parameter vector $\State$ to estimate is $13$: 12 coefficients in $\Coefficient$ and $\gamma \sim {\cal IG}(2,1.3)$ which corresponds to the unknown scale parameter of the EP distribution defined in Eq. (\ref{DefinitionExponentialPower}). We have chosen $\NbMCMCmove=5$ and $B=6$ for the adaptive MWG (summarized in Algo. \ref{Algo_AMWG}) used in the SMC sampler as forward kernel.  
	
Fig. \ref{fig:Chap4:CountData:FittingCurve} illustrates that the resulting mean prediction \Correct{curves} (and associated confidence \Correct{intervals}) obtained by using the proposed SMC sampler. The true curve is always within the confidence region which clearly shows the ability of our algorithm to give a good prediction of the functional relationship between the input and output variables. One advantage of Bayesian approach (vs optimization technique like LASSO) is the ability to provide a posterior confidence region for the Bayesian estimates.
	
	\begin{figure}[!h]
			\centering
			\includegraphics[width=0.65\textwidth]{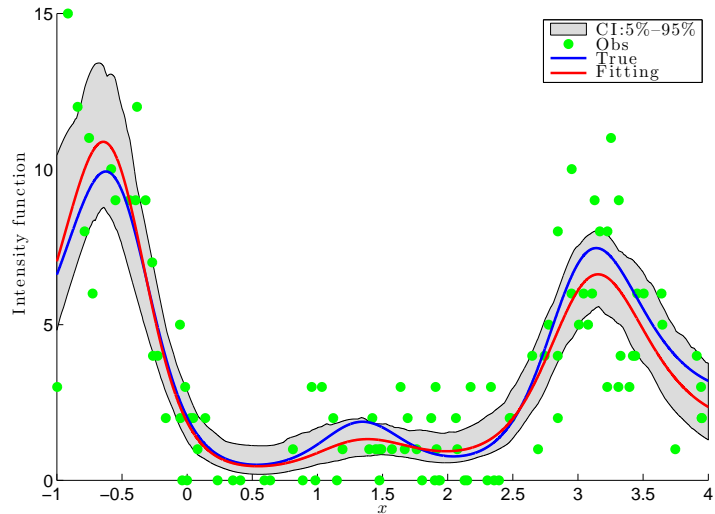}			
			\caption{Regression with count data [Prior: EP  $q=0.5$] by using Demix recycling scheme: true function in blue - observed responses in green circles - posterior mean from SMC sampler in red and confidence region in gray $5\%$ to $95\%$ percentiles. }
			\label{fig:Chap4:CountData:FittingCurve}
		\end{figure}
		
			Table \ref{Tab:Chap4:CountData:TrueModel_FixedY_Lq=0.5_pY}  shows the variance of the SMC \Correct{sampler} estimate of the normalizing constant of the target posterior distribution, $p({\bm y})$, obtained by using the proposed adaptive cooling schedule and the linear cooling schedule. As in previous examples, the variance obtained by using the proposed adaptive cooling schedule is significantly lower than the one obtained with a linear cooling schedule. This normalizing constant is a quantity of great interest for example whenever the best suitable basis function and/or distribution of the response have to be selected from a dictionary of different possible choices. The ability of the proposed adaptive cooling strategy to provide an estimator with smaller variance is clearly an important benefit in such cases.

		\begin{table}[!h]
			\centering
			\begin{footnotesize}
			\hspace*{-1cm}
			\begin{tabular}{ll|rcr|rcr}		
											&		&	 \multicolumn{3}{c|}{	Linear cooling schedule	}		&	\multicolumn{3}{c}{	Proposed Adaptive }      			\\	
											& & & & &\multicolumn{3}{c}{cooling schedule}      			\\	\hline 
											&	$\NbParticles=50$	&	-221.0070	&	$\pm$	&	124.3479	&	-203.3146	&	$\pm$	&	0.8215	\\		
			$\NbIterSMC=50$	&	$\NbParticles=200$	&	-211.8227	&	$\pm$	&	34.1198	&	-203.2687	&	$\pm$	&	0.2325	\\	\hline	
											&	$\NbParticles=50$	&	-211.4072	&	$\pm$	&	28.2070	&	-203.1100	&	$\pm$	&	0.4598	\\		
			$\NbIterSMC=100$	&	$\NbParticles=200$	&	-206.1121	&	$\pm$	&	10.0980	&	-203.1244	&	$\pm$	&	0.0698	\\	\hline	
											&	$\NbParticles=50$	&	-206.3567	&	$\pm$	&	13.6623	&	-202.9167	&	$\pm$	&	0.1627	\\		
			$\NbIterSMC=200$	&	$\NbParticles=200$	&	-204.1015	&	$\pm$	&	3.8083	&	-203.0268	&	$\pm$	&	0.0530	\\

			\end{tabular}
			\end{footnotesize}
			 \caption{The estimation of the marginal likelihood log $p({\bm y}|{\cal M}_1)$ (mean $\pm$ variance) in count data regression under model $\model_{1}$ [Prior: EP with $q=0.5$].}
			\label{Tab:Chap4:CountData:TrueModel_FixedY_Lq=0.5_pY}
\end{table}

		We now investigate the performance of the SMC sampler to correctly estimate the unknown coefficients of regression and as a consequence give some accurate prediction of the functional relationship  between the input and output variables. Table \ref{Tab:Chap4:CountData:TrueModel_Fixed_Lq_VarFittingCurve} clearly demonstrates that  our proposed scheme (ESS and DeMix) outperforms \Correct{ the two other schemes} that were used in these simulations  in terms of the stability of the posterior mean estimator. Again, the DeMix recycling scheme achieved a slightly better result than the ESS-based strategy in terms of variance.

		\begin{table}[!htb]
			\centering
			\begin{footnotesize}
			\hspace*{-1cm}
				\begin{tabular}{ll|c|c|c|c}
													&		&	No Recycling	&	Na\"ive	&	ESS-based	&	DeMix	\\	
													&		&		&	Recycling	&	Recycling	&	Recycling	\\	\hline 
													&	$\NbParticles=50$	&	0.00318883	&	0.00318883	&	0.00216641	&	0.00193355	\\		
					$\NbIterSMC=50$	&	$\NbParticles=200$	&	0.00069063	&	0.00069063	&	0.00050476	&	0.00046064	\\	\hline	
													&	$\NbParticles=50$	&	0.00268185	&	0.00268197	&	0.00148822	&	0.00130908	\\		
					$\NbIterSMC=100$	&	$\NbParticles=200$	&	0.00079279	&	0.00079267	&	0.00048881	&	0.00040874	\\	\hline
													&	$\NbParticles=50$	&	0.00308443	&	0.00307905	&	0.00170143	&	0.00149617	\\		
					$\NbIterSMC=200$	&	$\NbParticles=200$	&	0.00094644	&	0.00094534	&	0.00051125	&	0.00041064	\\		
				\end{tabular}
			\end{footnotesize}
			\caption{Variance of approximated curve  in count data regression  [Prior: EP with $q=0.5$].}
			\label{Tab:Chap4:CountData:TrueModel_Fixed_Lq_VarFittingCurve}
		\end{table}

\section{Conclusion}\label{ConclusionPaper}
In this paper, we discuss the use of SMC samplers for Bayesian inference. A \Correct{simple form of} the asymptotic variances for the SMC sampler \Correct{estimator} is derived under some assumptions. From this expression, a novel criterion to optimize is described in order to automatically and adaptively decide the cooling schedule of the algorithm. The proposed strategy consists in finding the evolution of the temperature along the SMC iterations that will optimize the (asymptotic) variance of the estimator of the normalizing constant of the target distribution. Furthermore, we propose two different approaches (ESS and DeMix) that recycle all past simulated particles for the final approximation of the posterior distribution. Numerical simulations clearly show that significant improvement can be obtained by using these different propositions into the SMC samplers.

\appendix[Proof of Proposition \ref{Proposition_Convergence_RB}]
\label{Proof_Proposition_Convergence_RB}

In this appendix, we present the proof of Proposition \ref{Proposition_Convergence_RB} related to the asymptotic variance of the SMC sampler estimator when resampling is performed before the sampling step. In \cite{DelMoral2006}, the authors does not study this case since the resampling cannot always be done before the sampling. In particular, as discussed in Section \ref{AlgorithmSettingsChap1} we can do the resampling before the sampling when the weights does not depend on the current value of the particle as it is the case when the backward kernel is the one used in this proposition.

\subsection{\Correct{On the estimation of an expectation}}

This results is quite straightforward to obtain by using classical Monte-Carlo results since we use a perfectly mixing kernel, the particles are (asymptotically) drawn at the $\TimeIndex$-th iteration, for $i=1,\ldots,N$:
\begin{equation}
\state_\TimeIndex^{(i)} \iid \pi_{\TimeIndex}(\cdot)
\end{equation}
which leads to the following particle estimate of the expectation:
\begin{equation}
\Expec_{\pi_\TimeIndex^N} (\varphi) = \frac{1}{N} \sum\limits_{i=1}^{N} \varphi(\state^{(i)}_{\TimeIndex})
\end{equation}
All particles are equally weighted since we have performed the resampling before the sampling step. As a consequence, we obtain:
\begin{equation}
							N^{\frac{1}{2}}\left \{\Expec_{\pi^{N}_{\TimeIndex}} (\varphi)-\Expec_{\pi_{\TimeIndex}}(\varphi)\right \}\Rightarrow \mathcal{N} (0,\sigma_{SMC,\TimeIndex}^{2}(\varphi))
						\end{equation}	
						with $\sigma_{SMC,\TimeIndex}^{2}(\varphi)$ the variance of $\varphi(\state)$ with respect to $\pi_{\TimeIndex}$, i.e. $\sigma_{SMC,\TimeIndex}^{2}(\varphi)=\Var_{\pi_{\TimeIndex}}(\varphi(\state))$.

 \subsection{\Correct{On the estimation of the normalizing constant}}

In this section, we will derive the asymptotic variance related to the estimator of the normalizing constant. Let us firstly study the estimate of the ratio of normalizing constant,$\NormConst_{\TimeIndex }/\NormConst_{\TimeIndex-1 }$, defined in Eq. (\ref{Eq_Approx_Ratio_NormalizingConstant}) which  is given in the context of the proposition \ref{Proof_Proposition_Convergence_RB} by:
\begin{eqnarray}\label{EstimatorNormProof3}
\widehat{\dfrac{\NormConst_{\TimeIndex }}{\NormConst_{\TimeIndex -1}}}&=&\sum\limits _{\ParsIndexSMC =1}^{\NbParticles } \NormISWeight  _{\TimeIndex -1}^{(\ParsIndexSMC )} \UnNorIncreWeight _{\TimeIndex } (\State _{\TimeIndex -1}^{(\ParsIndexSMC )},\State _{\TimeIndex }^{(\ParsIndexSMC )})=\frac{1}{\NbParticles} \sum\limits _{\ParsIndexSMC =1}^{\NbParticles }  \frac{\UnNorTarget _{\TimeIndex}(\state  _{\TimeIndex -1}^{(\ParsIndexSMC )})}{\UnNorTarget _{\TimeIndex -1}(\state  _{\TimeIndex -1}^{(\ParsIndexSMC )})}
\end{eqnarray}
since the particles are equally weighted due to the resampling  before the sampling and the unnormalized incremental weights are defined in Eq. (\ref{Eq_UnNormalized_IncrementWeights1}) when the backward kernel in Eq. (\ref{Eq_SubOptimal_BackwardKernel}) is used. Moreover, owing to the perfect mixing assumption, we have:
 for $i=1,\ldots,N$:
\begin{equation}\label{ParticleDistribProof3}
\state_{\TimeIndex-1}^{(i)} \iid \pi_{\TimeIndex-1}(\cdot)
\end{equation}

From (\ref{EstimatorNormProof3}) and (\ref{ParticleDistribProof3}), the unbiasedness of this estimator is obvious:
\begin{eqnarray}\label{ProofUnbiasedBeforeResampling}
\Expec_{\pi_{\TimeIndex-1}} \left[\widehat{\dfrac{\NormConst_{\TimeIndex }}{\NormConst_{\TimeIndex -1}}} \right] & = & \int  \frac{\UnNorTarget _{\TimeIndex}(\state  _{\TimeIndex -1})}{\UnNorTarget _{\TimeIndex -1}(\state  _{\TimeIndex -1})} \pi_{\TimeIndex-1}(\state  _{\TimeIndex -1}) d\state  _{\TimeIndex -1}  \\
& = & \int  \frac{\UnNorTarget _{\TimeIndex}(\state  _{\TimeIndex -1})}{\UnNorTarget _{\TimeIndex -1}(\state  _{\TimeIndex -1})} \frac{\UnNorTarget _{\TimeIndex -1}(\state  _{\TimeIndex -1})}{\NormConst_{\TimeIndex -1}} d\state  _{\TimeIndex -1} =  \dfrac{\NormConst_{\TimeIndex }}{\NormConst_{\TimeIndex -1}} \nonumber
\end{eqnarray}
Let us now study the variance of this estimator:
\begin{eqnarray}
\Var \left(\widehat{\dfrac{\NormConst_{\TimeIndex }}{\NormConst_{\TimeIndex -1}}}\right)& = & \frac{1}{N} \sum\limits _{\ParsIndexSMC =1}^{\NbParticles }  \Var\left(\frac{\UnNorTarget _{\TimeIndex}(\state  _{\TimeIndex -1})}{\UnNorTarget _{\TimeIndex -1}(\state  _{\TimeIndex -1})}\right) \\
& = & \frac{1}{N} \left\{ \Expec_{\pi_{\TimeIndex-1}} \left[ \frac{\UnNorTarget^2 _{\TimeIndex}(\state  _{\TimeIndex -1})}{\UnNorTarget ^2_{\TimeIndex -1}(\state  _{\TimeIndex -1})}\right] - \Expec^2_{\pi_{\TimeIndex-1}} \left[ \frac{\UnNorTarget _{\TimeIndex}(\state  _{\TimeIndex -1})}{\UnNorTarget _{\TimeIndex -1}(\state  _{\TimeIndex -1})}\right] \right\} \nonumber
\end{eqnarray}
In this expression, the mean has already been derived in Eq. (\ref{ProofUnbiasedBeforeResampling}) and the second moment can be written as
\begin{eqnarray}
\Expec_{\pi_{\TimeIndex-1}} \left[ \frac{\UnNorTarget^2 _{\TimeIndex}(\state  _{\TimeIndex -1})}{\UnNorTarget ^2_{\TimeIndex -1}(\state  _{\TimeIndex -1})}\right] & = & \int \pi_{\TimeIndex-1}(\state  _{\TimeIndex -1}) \frac{\UnNorTarget^2 _{\TimeIndex}(\state  _{\TimeIndex -1})}{\UnNorTarget ^2_{\TimeIndex -1}(\state  _{\TimeIndex -1})} d\state  _{\TimeIndex -1} \nonumber \\
& = & \dfrac{\NormConst^2_{\TimeIndex }}{\NormConst^2_{\TimeIndex -1}} \int \frac{\pi_{\TimeIndex}^2(\state  _{\TimeIndex -1})}{\pi_{\TimeIndex-1}(\state  _{\TimeIndex -1})}d\state  _{\TimeIndex -1}
\end{eqnarray}
which give the following expression for the variance:
\begin{equation}\label{ProofVarianceRatioNormBR}
\Var \left(\widehat{\dfrac{\NormConst_{\TimeIndex }}{\NormConst_{\TimeIndex -1}}}\right)=\frac{1}{N} \left( \dfrac{\NormConst_{\TimeIndex }}{\NormConst_{\TimeIndex -1}} \right)^2 \left[ \int \frac{\pi_{\TimeIndex}^2(\state  _{\TimeIndex -1})}{\pi_{\TimeIndex-1}(\state  _{\TimeIndex -1})}d\state  _{\TimeIndex -1} -1 \right]
\end{equation}

In the results given in Proposition \ref{Proof_Proposition_Convergence_RB}, we want to have the variance of the log of the normalizing constant at time $\TimeIndex$ which can be rewritten using Eq. (\ref{Eq_Approx_Ratio_NormalizingConstantZ0Zt}) as
\begin{equation}\label{ProofDecompoWithSumLog}
\log\left (\widehat{\dfrac{\NormConst_{\TimeIndex}}{\NormConst_{1}}}\right ) = \sum_{n=2}^{\TimeIndex} \log \left (\widehat{\dfrac{\NormConst_{n}}{\NormConst_{n-1}}}\right ) 
\end{equation}
From this expression, we have to obtain the variance of the $\log$ ratio of the normalizing constant. This term can be obtained by using the \textit{delta method}    \cite{BergerStatisticalInference} that states that if
\begin{equation}
\NbParticles^{\frac{1}{2}} \left(X_n- \mu \right) \Rightarrow \mathcal{N} (0,\sigma^{2})
\end{equation}
then for a given function $g$ and a specific value of $\mu$ (by assuming that $g'(\mu)$ exists and is not 0)
\begin{equation}
\NbParticles^{\frac{1}{2}} \left(g(X_n)- g(\mu) \right) \Rightarrow \mathcal{N} (0,\sigma^{2} \left[g'(\mu) \right]^2)
\end{equation}

By using this delta method and Eqs. (\ref{ProofVarianceRatioNormBR}) and (\ref{ProofDecompoWithSumLog}), we finally obtain the result presented in Proposition  \ref{Proposition_Convergence_RB}, i.e.:
\begin{equation}
							N^{\frac{1}{2}} \left \{ \log\left (\widehat{\dfrac{\NormConst_{\TimeIndex}}{\NormConst_{1}}}\right )-\log \left (\dfrac{\NormConst_{\TimeIndex}}{\NormConst_{1}}\right ) \right \} \Rightarrow \mathcal{N} (0,\sigma^{2}_{SMC,\TimeIndex})
						\end{equation}
						
						with 
						\begin{equation}
							 	\begin{aligned}
							 		\sigma^{2}_{SMC,\TimeIndex}&=\int \dfrac{\pi_{2}^{2}(\state_{1})}{\eta_{1}(\state_{1})}d\state_{1}&+\sum\limits_{k=2}^{\TimeIndex-1} \int \dfrac{\pi_{k+1}^{2}(\state_{k})}{\pi_{k}(\state_{k})}d\state_{k} -(\TimeIndex-1)
							 	\end{aligned}
						\end{equation}

\bibliographystyle{IEEEtran}
\bibliography{references,bibliography_Tracking,bibliography_Thesis,bibliography_Regression,bibliography_OtherChap2,referencesGWP}

\end{document}